\documentclass[12pt,dvipdfmx,superscriptaddress,nofootinbib]{revtex4-1}

\usepackage{amsmath,amssymb,amsthm}
\usepackage{graphicx}
\usepackage{bm}
\usepackage{cases}
\usepackage{tikz}
\usepackage{color}
\usepackage{ulem}

\newcommand{\dd}{\mathrm{\bf d}}

\newtheorem{defn}{Definition}
\newtheorem{prop}{Proposition}
\newtheorem{thm}{Theorem}
\newtheorem{lem}{Lemma}

\definecolor{DarkBlue}{rgb}{0,0,0.7} 

\definecolor{DarkRed}{rgb}{0.65,0,0}

\begin{document}

\baselineskip5.5mm
\thispagestyle{empty}

{\baselineskip0pt
\small
\leftline{\baselineskip16pt\sl\vbox to0pt{
               \hbox{\it Division of Particle and Astrophysical Science, Nagoya University}
               \hbox{\it Department of Physics, Rikkyo University} 
                             \vss}}
\rightline{\baselineskip16pt\rm\vbox to20pt{
            {
            \hbox{RUP-20-8}
            }
\vss}}%
}

\author{Masataka~Tsuchiya}\email{tsuchiya.masataka@h.mbox.nagoya-u.ac.jp}
\affiliation{
Division of Particle and Astrophysical Science,
Graduate School of Science, Nagoya University, 
Nagoya 464-8602, Japan
}

\author{Chul-Moon~Yoo}\email{yoo@gravity.phys.nagoya-u.ac.jp}
\affiliation{
Division of Particle and Astrophysical Science,
Graduate School of Science, Nagoya University, 
Nagoya 464-8602, Japan
}

\author{Yasutaka~Koga}\email{koga@rikkyo.ac.jp}
\affiliation{
Department of Physics, Rikkyo University, Toshima,
Tokyo 171-8501, Japan
}

\author{Tomohiro~Harada}\email{harada@rikkyo.ac.jp}
\affiliation{
Department of Physics, Rikkyo University, Toshima,
Tokyo 171-8501, Japan
}

\vskip1cm

\title{Sonic Point and Photon Surface}


\begin{abstract}
\vspace{5mm}
The sonic point/photon surface correspondence is thoroughly investigated in a general setting. 
First, we investigate a sonic point of a transonic steady perfect fluid flow 
in a general stationary spacetime, particularly focusing on the radiation fluid. 
The necessary conditions that the flow must satisfy at a sonic point
 are derived as conditions for the kinematical quantities of the congruence of streamlines
in analogy with the de Laval nozzle equation in fluid mechanics. 
We compare the conditions for a sonic point with the notion of a photon surface, which can be defined as a timelike totally umbilical hypersurface. 
As a result, we find that, for the realization of the sonic point/photon surface correspondence, 
the speed of sound $v_{\rm s}$ must be given by $1/\sqrt{d}$ with $d$ being the spatial dimension of the spacetime.
For the radiation fluid ($v_{\rm s}=1/\sqrt{d}$), we confirm that 
a part of the conditions is shared by the sonic point and the photon surface. 
However, 
in general, a Bondi surface, a set of sonic points, does not necessarily coincide 
with a photon surface. 
Additional assumptions, such as a spatial symmetry, are essential to the realization of the sonic point/photon surface correspondence 
in all known examples.  
\end{abstract}

\maketitle

\pagebreak

\section{introduction}
\label{I}

In recent years, strong relativistic gravitational fields play an important role 
not only in astrophysics and cosmology but also in high energy physics. 
One very useful approach to characterize the gravitational field is 
to consider a probe matter, such as test particles, on the background gravitational field. 
For example, the photon sphere is introduced by a family of unstable circular orbits of test massless particles, 
and it is responsible for the radius of a black hole shadow, the silhouette in optical observation.
Specifically, it is well known that a photon sphere in the Schwarzschild spacetime of mass $M$ 
is given by the cylindrical hypersurface of radius $3M$. 
The extensions of the photon sphere to more general situations and those 
mathematical aspects have 
been actively 
studied~\cite{Cederbaum:2014gva,Cederbaum:2015aha,Cederbaum:2015fra,Shiromizu:2017ego,Yoshino:2017gqv}, 
and the notion of a photon surface 
has been proposed as a generalization of the photon sphere 
to other topologies~\cite{Claudel:2000yi,Gibbons:2016isj}.
Another important example of the probe matter 
is the perfect fluid 
on the background spacetime. 
An astrophysical transonic flow, such as gas accretion and relativistic jet, is responsible for energy transfer. 
As for the accretion, we usually consider a transonic steady flow of a perfect fluid. 
The simplest model would be a Bondi-type flow, 
a steady spherical accretion flow~\cite{10.1093/mnras/112.2.195,Moncrief:1980,Babichev:2005nc,Roy:2007xf,Ahmed:2016cuy,Rodrigues:2016uor}. 
Models describing a rotating gas accretion to form a transonic disk 
have also been considered~\cite{Shakura:1972te,Parev:1995me,Beskin:2002ux,Mukhopadhyay:2008ge}.

Related to the above two examples,  the following theorem 
called the {\it sonic point/photon sphere correspondence} has been given in Ref.~\cite{Koga:2016jjq}:
\begin{thm}[Sonic point/photon sphere correspondence for steady spherical flow~\cite{Koga:2016jjq}]
Consider a transonic steady radial flow of radiation fluid and 
a static observer of the flow in a general static spherically symmetric spacetime 
of arbitrary dimension. 
If the flow is transonic, 
the radius of a sonic point coincides with that of an unstable or marginally stable photon sphere 
of constant radius (see Definition~\ref{def:1} in Sec.~\ref{IIIc} 
for the definition of the stability of a photon sphere).
\label{thm:spps}
\end{thm}
Theorem~\ref{thm:spps} includes the results reported in Refs.~\cite{Mach:2013gia,Ficek:2015eya}, 
where some specific static spherically symmetric spacetimes are considered. 
The universality of the coincidence of the photon sphere and the Bondi surface, the set of the sonic points, 
is pointed out by M.~Cveti\v{c}, G.W.~Gibbons and C.N.~Pope~\cite{Cvetic:2016bxi} in 2016 
independently of Ref.~\cite{Koga:2016jjq}.
After that, two of the present authors extended the correspondence 
to a rotating thin disk model in the same geometry in Ref.~\cite{Koga:2018ybs}, and 
spatially planar and hyperbolic cases are considered in Ref.~\cite{Koga:2019teu}.
The result in Ref.~\cite{Koga:2019teu} shows that the spatial topology no longer matters, and therefore 
we call the coincidence {\it sonic point/photon surface correspondence}.

It is remarkable that we can also define the photon surface 
by using the trace-free part of the second fundamental form as follows:
\begin{thm}[\cite{Claudel:2000yi,Perlick:2005jn,Koga:2019uqd}]
A timelike hypersurface $S$ immersed in a spacetime $(M,\bm g)$ of arbitrary dimension is a photon surface 
if and only if the hypersurface is totally umbilical, that is, 
for every $p\in S$ and any vectors $\bm X,\bm Y\in T_{p}S$,
\begin{equation}
\bm\sigma_{\bm \chi}(\bm X,\bm Y)=0,
\label{umb3}
\end{equation}
where $\bm\sigma_{\bm \chi}$ is the trace-free part of the second fundamental form $\bm\chi$ for $S$.
\label{def:2}
\end{thm}
Theorem~\ref{thm:spps} indicates that the shear tensor for the congruence of 
streamlines may specify the locus of the sonic point as well as 
the photon surface. 
As pointed out in Ref.~\cite{Cvetic:2016bxi}, this idea is compatible with the insight in the de Laval nozzle model, 
in which the sonic point is located at the throat, where the cross sectional area is minimal. 

One of our purposes in this paper is to reveal the 
role of the shear tensor for the congruence of streamlines 
in specifying the locus of the sonic point. 
We start this paper with the existence of a timelike Killing vector field. 
A stationary flow of perfect fluid and a fiducial observer are defined 
in terms of the timelike Killing vector field. The speed of the fluid flow is 
defined by the fluid velocity relative to the fiducial observer.
We derive the necessary conditions the flow must satisfy at the sonic point in this model. 
The conditions are obtained as an equation and an inequality which are conventionally 
derived 
through the phase space analysis (see, e.g. Ref.~\cite{Chaverra:2015bya}). 
It is remarkable that the necessary conditions for the radiation fluid 
reduce to the conditions imposed on the shear tensor for the congruence of streamlines. 
Therefore, 
we consider that the shear tensor is essential in the sonic point/photon surface correspondence, 
and revisit the sonic point/photon surface correspondence in terms 
of the necessary conditions imposed on 
a timelike hypersurface associated with the flow.

This paper is organized as follows.
In Sec.~\ref{II}, we introduce the notion of the steady flow, and 
rewrite the basic equations of a steady perfect fluid flow 
to derive the equation for the congruence of the streamlines 
similarly to the de Laval nozzle system. 
We evaluate the de Laval nozzle-like equation 
at a sonic point, 
and 
obtain one equation and one inequality 
that the flow must satisfy at the sonic point. 
In particular, for the radiation fluid, the necessary conditions are 
rewritten in terms of the time-time component of the shear tensor for the congruence of the streamlines.
In Sec.~\ref{III}, we introduce the notion of the proper 
section of a congruence of streamlines as the corresponding hypersurface to the flow at each point on a stream line.
It is actually the relativistic extension of the section of the de Laval nozzle model. 
We can then rewrite the necessary conditions for the flow
in terms of  
the trace-free part of the second fundamental form for the proper section. 
Finally, in Sec.~\ref{IIIc}, we apply our observations of the steady perfect fluid flow 
to the sonic point/photon surface correspondence. 
Sec.~\ref{IV} is devoted to 
a summary.

We use geometrized units in which both the speed of light $c$ and Newton's
gravitational constant $G$ are one.

\section{Sonic point of steady radiation fluid flow equipped with Killing observers}
\label{II}

In the systems~\cite{Koga:2016jjq,Koga:2018ybs,Koga:2019teu} 
where the sonic point/photon surface correspondence was reported, 
a static observer or a co-rotating observer with the flow was considered.
One can understand these choices of the observer are based on timelike isometries of the spacetime, 
i.e. the observer is chosen so that the world line will be the orbit of a timelike Killing vector field.
From this view point, 
we begin this paper with a definition of a 
steady perfect fluid flow 
in a general stationary spacetime of arbitrary dimension. 
We suppose that the spacetime is smooth enough, and 
do not consider any discontinuity of fluid variables. 
Specifically, we
simply assume at least $C^{2}$ differentiability 
for fluid variables in the vicinity of the point under consideration. 
In this section, we derive two necessary conditions that the flow must satisfy at the sonic point. 
Those conditions will be imposed on the time-time component of the shear tensor for the congruence of streamlines if we consider the radiation fluid.

\subsection{Steady perfect fluid flow}
\label{II1}

Let us consider a perfect fluid flow in a $(d+1)$-dimensional spacetime $(M,\bm g)$ 
described by thermodynamic variables $M\to\mathbb{R}_{>0}$: number density $n$, specific enthalpy $h$, pressure $P$, 
specific entropy $s$ and temperature $T$, together with fluid $(d+1)$-velocity $\bm u$ 
which satisfies the normalization condition
\begin{equation}
\bm u\cdot \bm u=-1,
\label{norm}
\end{equation}
where the centered dot ``~$\,\cdot\,$~'' denotes the inner product with respect to $\bm g$. 
The fluid obeys the first law of thermodynamics, 
the continuity equation, the energy-momentum conservation law and the equation of state:
\begin{subequations}
\begin{gather}
\dd h=T\dd s+n^{-1}\dd P, \label{1st}  \\
\nabla\cdot(n\bm u)=0, \label{cont} \\ 
\nabla\cdot\left(nh\bm u\otimes \bm u+P\bm g^{-1}\right)=0, \label{energy} \\
h=h(P,s). \label{eos}
\end{gather}
\end{subequations}
It is well known that Eqs.~\eqref{norm}, \eqref{1st}, \eqref{cont} and Eq.~\eqref{energy} contracted with $\bm u$ 
lead to
\begin{equation}
{\cal L}_{\bm u}s=0.
\label{ise}
\end{equation}
The functional form of the equation of state~\eqref{eos} needs to be restricted in this paper 
so that Eq.~\eqref{eos} will properly define the speed of sound $v_{\rm s}:M\to(0,1)$ by 
\begin{equation}
v_{\rm s}^{2}:=\left(\frac{\partial\ln h}{\partial\ln n}\right)_{s}.
\label{vs}
\end{equation}
The definition~\eqref{vs} of the speed of sound is also expressed as
$v_{\rm
s}^{2}=(\partial P/\partial\rho)_{s}$ using the energy density $\rho$,
and therefore we have supposed that Eq.~\eqref{eos} gives a strictly
monotonically increasing function $P(\rho,s)$ 
of $\rho$ if $s$ is fixed.
The $d$-velocity of the flow $v_{\rm obs}:M\to[0,1)$ with respect to an observer $\bm u_{\rm obs}$ is defined by the orthogonal decomposition
\begin{equation}
\bm u=\frac{1}{\sqrt{1-v_{\rm obs}^{2}}}\left(\bm u_{\rm obs}+v_{\rm obs}\bar{\bm w}\right),
\label{decomp:obs}
\end{equation}
where $\bm u_{\rm obs}$ and $\bar{\bm w}$ are the vector fields on $M$ satisfying
\begin{subequations}
\begin{eqnarray}
\bm u_{\rm obs}\cdot\bm u_{\rm obs}&=&-1,\\
\bar{\bm w}\cdot\bar{\bm w}&=&1,\\
\bm u_{\rm obs}\cdot\bar{\bm w}&=&0.
\end{eqnarray}
\end{subequations}
The integral curves of $\bm u_{\rm obs}$ are regarded as the world lines of the observers.

Throughout this paper, we assume the existence of a timelike Killing vector field $\bm \xi$ 
 in a neighborhood $U$ of a point $p\in M$ under consideration, and the observer $\bm u_{\rm obs}$ is supposed to satisfy
\begin{equation}
(\bm u_{\rm obs}-\bar{\bm\xi})|_{p}=0,
\end{equation}
where $\bar{\bm\xi}:=\bm\xi/\sqrt{|\bm\xi\cdot\bm\xi|}$ is defined in the neighborhood $U\ni p$.
Then we define the steady perfect fluid flow as follows: 
\begin{defn}[Steady perfect fluid flow]
We say that the perfect fluid flow is steady if
there exists a timelike Killing vector field $\bm\xi$ on $U$ 
which, at every point $p\in U$ on the fluid flow, 
satisfies the following equations:
\begin{subequations}
\begin{eqnarray}
{\cal L}_{\bm\xi}P={\cal L}_{\bm\xi}s&=&0,\label{steady1}\\
{\cal L}_{\bm\xi}\bm u&=&0. \label{steady2}
\end{eqnarray}
\end{subequations}
\end{defn}

Our setting 
is consistent with most models of accretion flows onto gravitational sources, 
where  one supposes that the flow is in steady state, 
and the observer vector field $\bm u_{\rm obs}$ is tangent to a timelike Killing vector field on 
the background spacetime.
Note that $\bm\xi$ may not be given as $\bm\xi=\bm\partial_{t}$ 
in the standard coordinate system of a stationary spacetime.
For example, we will take 
$\bm\xi=\bm\partial_{t}+\Omega\bm\partial_{\phi}$ in Sec.~\ref{IIIc}.

Let us focus on a single point $p$.
The Killing vector field $\bm\xi$ obeys
\begin{subequations}
\begin{gather}
{\cal L}_{\bm\xi}\bm g=0,\label{killing}\\
\bm\xi\cdot\bm\xi\,|_{U}<0.
\end{gather}
\end{subequations}
The stationarity~\eqref{steady1} of $P$ and $s$ immediately leads to
\begin{equation}
{\cal L}_{\bm\xi}n={\cal L}_{\bm\xi}h={\cal L}_{\bm\xi}T=0.
\end{equation}
In the same way as Eq.~\eqref{decomp:obs}, we define the speed of the flow $v:U\to[0,1)$ 
by the orthogonal decomposition of $\bm u$ as follows:
\begin{equation}
\bm u=\frac{1}{\sqrt{1-v^{2}}}\left(\bar{\bm\xi}+v\bar{\bm\eta}\right),
\label{decomp}
\end{equation}
where $\bar{\bm\eta}$ is a spacelike unit vector field on $U$ orthogonal to $\bar{\bm\xi}$, that is,  
\begin{subequations}
\begin{eqnarray}
\bar{\bm\xi}\cdot\bar{\bm\xi}&=&-1,\\
\bar{\bm\eta}\cdot\bar{\bm\eta}&=&1,\label{norm:eta}\\
\bar{\bm\xi}\cdot\bar{\bm\eta}&=&0.
\end{eqnarray}
\end{subequations}
Taking the inner product of Eq.~\eqref{decomp} with $\bar{\bm\xi}$ or $\bar{\bm\eta}$, we obtain the following expressions of $v$:
\begin{equation}
v^{2}~=~1-\frac{1}{(\bm u\cdot \bar{\bm \xi}\,)^{2}}~=~\frac{(\bm u\cdot\bar{\bm\eta})^{2}}{1+(\bm u\cdot\bar{\bm\eta})^{2}}. 
\label{v}
\end{equation}
We can also show the invariance of $\bar{\bm\xi}$, $\bar{\bm\eta}$ and $v$ under ${\cal L}_{\bm\xi}$ from Eqs.~\eqref{killing}, \eqref{steady1} and \eqref{steady2}.
The Killing equation~\eqref{killing} immediately yields
\begin{equation}
{\cal L}_{\bm\xi}\bar{\bm\xi}=0.
\end{equation}
From the first expression of $v^{2}$ in Eq.~\eqref{v}, we obtain 
\begin{equation}
{\cal L}_{\bm\xi}v^{2}=0.
\end{equation}
Acting ${\cal L}_{\bm\xi}$ on Eq.~\eqref{decomp} leads to
\begin{equation}
v{\cal L}_{\bm\xi}\bar{\bm\eta}=0.
\end{equation}
Here, we note that $\bar{\bm\eta}$ is not uniquely given from Eq.~\eqref{decomp} for $v=0$. 
However, by choosing an appropriate $\bar{\bm\eta}|_{\{v=0\}}$, we always obtain $\bar{\bm\eta}$ such that
\begin{equation}
{\cal L}_{\bm\xi}\bar{\bm\eta}=0.
\label{frob}
\end{equation}
Eq.~\eqref{frob} is useful not only 
for rewriting the basic equations, but it is also essential in the description of the phase space analysis in Appendix~\ref{app:phase}, which requires a certain coordinate system of 2-dimension on the foliation of $U$.

\subsection{Congruence of streamlines}

In the previous subsection, we introduced the spacelike vector field $\bar{\bm\eta}$ orthogonal to $\bar{\bm\xi}$.
Actually, the integral curves of $\bar{\bm\eta}$ are what we usually call the streamlines, and we show the locus of the sonic point is determined by the shear tensor of the streamlines in this paper.
As the first step, we introduce the tensor field
\begin{equation}
\bm B:=\nabla\otimes\bar{\bm\eta}^{\flat}=\nabla_{\mu}\bar{\eta}_{\nu}\,\dd x^{\mu}\otimes\dd x^{\nu}
\label{B}
\end{equation}
describing the congruence of streamlines, where the flat ``~$^{\flat}$~'' denotes the covariant dual of vector fields with respect to $\bm g$.
$\bm B$ obeys the following equations which will be used throughout this paper.
Acting $\nabla_{\mu}$ on the normalization condition~\eqref{norm:eta} of  $\bar{\bm\eta}$ gives
\begin{equation}
\bm B(\,\cdot\,,\bar{\bm\eta})=0.
\label{B:1}
\end{equation}
Eq.~\eqref{frob} is deformed to
\begin{equation}
\bm B(\bm\xi,\,\cdot\,)-\nabla_{\bar{\bm\eta}}\bm\xi=0.
\label{B:2}
\end{equation}
Applying the Killing equation $\nabla_{\mu}\xi_{\nu}+\nabla_{\nu}\xi_{\mu}=0$ to the second term of Eq.~\eqref{B:2}, 
we obtain 
\begin{equation}
-\nabla_{\bar{\bm\eta}}\bm\xi=(\nabla_{\mu}\xi^{\nu})\bar{\eta}_{\nu}=-\bm B(\,\cdot\,,\bm\xi).
\end{equation}
Therefore, Eq.~\eqref{B:2} results in
\begin{equation}
\bm B(\bar{\bm\xi},\,\cdot\,)-\bm B(\,\cdot\,,\bar{\bm\xi})=0.
\label{B:3}
\end{equation}

In the following, we shall perform the usual unique decomposition of $\bm B$.
Because the streamlines are not necessarily geodesics, $\bm B$ includes the acceleration
defined by 
\begin{equation}
\bm a:=\bm B(\bar{\bm\eta},\,\cdot\,)^{\sharp}=\nabla_{\bar{\bm\eta}}\bar{\bm\eta},
\label{a}
\end{equation}
where the sharp ``\,$\sharp$\,'' denotes the covariant dual of covector fields with respect to $\bm g$.
Taking the inner product of Eq.~\eqref{B:1} and \eqref{B:3} with $\bar{\bm\eta}$, respectively, 
we obtain the orthogonal relations of $\bm a$ with $\bar{\bm\xi}$ and $\bar{\bm\eta}$:
\begin{subequations}
\begin{eqnarray}
\bm a\cdot\bar{\bm\xi}&=&0,\label{a xi}\\
\bm a\cdot\bar{\bm\eta}&=&0.
\label{a eta}
\end{eqnarray}
\end{subequations}
Eq.~\eqref{a xi} states that $\bm a$ is orthogonal to a timelike vector field, i.e. $\bm a$ is spacelike.
Let us denote the projection of tensor fields with respect to $\bar{\bm\eta}$ by $\perp$.
The projection tensor is given by
\begin{equation}
\bm g^{\perp}:=\bm g-\bar{\bm\eta}^{\flat}\otimes\bar{\bm\eta}^{\flat}.
\end{equation}
Taking Eqs.~\eqref{B:1} and \eqref{a} into account, we obtain the projection of $\bm B$ as 
\begin{equation}
\bm B^{\perp}:=\bm B-\bar{\bm\eta}^{\flat}\otimes\bm a^{\flat}.
\label{B:perp1}
\end{equation}
$\bm B^{\perp}$ consists of the trace $\Theta$, the trace-free symmetric part $\bm\sigma$ and the 
anti-symmetric part $\bm\omega$:
\begin{equation}
\bm B^{\perp}=\frac{\Theta}{d}\,\bm g^{\perp}+\bm\sigma+\bm\omega.
\label{B:perp2}
\end{equation}
$\Theta$, $\bm\sigma$ and $\bm\omega$ are called the expansion scalar, the shear tensor and the vorticity tensor, respectively.
Using Eq.~\eqref{a eta}, we rewrite $\Theta:={\rm tr}_{\bm g^{\perp}}\bm B^{\perp}$ as 
\begin{equation}
\Theta={\rm tr}_{\bm g}\bm B=\nabla\cdot\bar{\bm\eta}.
\label{Theta}
\end{equation}

\subsection{Sonic point, shear tensor and radiation fluid}
\label{IIc}

From the energy-momentum conservation law~\eqref{energy} contracted with $\bar{\bm\eta}$, 
we can derive the de Laval nozzle-like equation for a steady perfect fluid flow:
\begin{equation}
\left(v_{\rm s}^{2}-v^{2}\right){\cal L}_{\bar{\bm\eta}}\left(\ln\frac{|v|}{\sqrt{1-v^{2}}}\right)+\Phi=0,
\label{de laval}
\end{equation}
where 
\begin{equation}
\Phi=\bm B(\bar{\bm\xi},\bar{\bm\xi})+v_{\rm s}^{2}~\Theta.
\label{Phi0}
\end{equation}
We also use Eqs.~\eqref{1st}, \eqref{cont} and \eqref{eos} in the derivation of Eq.~\eqref{de laval}.
Readers may refer to Appendix~\ref{app:de laval} for the details of the derivation of Eq.~\eqref{de laval}.

The point $p$ is called the sonic point if the following equation is satisfied:
\begin{equation}
\left(v_{\rm s}^{2}-v^{2}\right)\Big|_{p}=0.
\label{sonic0}
\end{equation}
From the de Laval nozzle-like equation~\eqref{de laval}, we obtain two
necessary conditions that the steady perfect fluid flow must satisfy  at the sonic point as is shown 
in Theorem~\ref{prop:con} below.

\begin{thm}
If $p$ is a sonic point, the following two conditions must be satisfied:
\begin{subequations}
\begin{eqnarray}
\Phi|_{p}&=&0,
\label{con1}\\
{\cal L}_{\bar{\bm\eta}}\Phi|_{p}&\ge&-\frac{1}{2}v^{-2}(1-v^{2})^{-1}{\cal L}_{\bar{\bm\eta}}\left(v_{\rm s}^{2}-v^{2}\right){\cal L}_{\bar{\bm\eta}}v_{\rm s}^{2}\Big|_{p}.
\label{con2}
\end{eqnarray}
\end{subequations}
\label{prop:con}
\end{thm}
\begin{proof}
The first term in the de Laval nozzle-like equation~\eqref{de laval} vanishes at $p$.
Simultaneously the second term in Eq.~\eqref{de laval} must also vanish, and we obtain the first condition~\eqref{con1}.
In addition, acting ${\cal L}_{\bar{\bm\eta}}$ on both sides of Eq.~\eqref{de laval}, we get the following equation at $p$:
\begin{equation}
\left\{{\cal L}_{\bar{\bm\eta}}(v_{\rm s}^{2}-v^{2}){\cal L}_{\bar{\bm\eta}}\left(\ln\frac{|v|}{\sqrt{1-v^{2}}}\right)+{\cal L}_{\bar{\bm\eta}}\Phi\right\}\Big|_{p}=0,
\end{equation}
whence
\begin{equation}
\frac{1}{2}v^{-2}\left(1-v^{2}\right)^{-1}\left[{\cal L}_{\bar{\bm\eta}}\left(v_{\rm s}^{2}-v^{2}\right)\right]^{2}\Big|_{p}=\left\{{\cal L}_{\bar{\bm\eta}}\Phi+\frac{1}{2}v^{-2}\left(1-v^{2}\right)^{-1}{\cal L}_{\bar{\bm\eta}}\left(v_{\rm s}^{2}-v^{2}\right){\cal L}_{\bar{\bm\eta}}v_{\rm s}^{2}\right\}\Big|_{p}.
\end{equation}
Because the left-hand side is non-negative, the right-hand side must also be non-negative.
Therefore, we obtain the second condition~\eqref{con2}.
\end{proof}
The same conditions can be derived through the phase space analysis (see Appendix~\ref{app:phase}).
Note that the flow requires only $C^{1}$ smoothness at $p$ for Eq.~\eqref{con1}, while the flow requires $C^{2}$ smoothness at $p$ for the inequality~\eqref{con2}. 

It is remarkable that the function $\Phi$ which has been defined in Eq.~\eqref{Phi0} reduces to the $(\bar{\bm\xi},\bar{\bm\xi})$ component of the shear tensor $\bm\sigma$ 
at a point where the speed of sound is given by $1/\sqrt{d}$ with $d$ being the spatial dimension of the spacetime.
Combining this fact with Theorem~\ref{prop:con}, we arrive at the following propositions mentioning the importance of the shear tensor in specifying the locus of the sonic point:

\begin{prop}
Suppose that $p$ is a sonic point, 
and that the expansion $\Theta$ does not vanish at $p$.
Then the speed of sound at $p$ is given by $1/\sqrt{d}$ with $d$ being the spatial dimension of the spacetime if and only if the $(\bar{\bm\xi},\bar{\bm\xi})$ component of the shear tensor $\bm\sigma$ vanishes at $p$.
\label{prop:5}
\end{prop}
\begin{proof}
($\Longrightarrow$)
Assume that $v_{\rm s}|_{p}=1/\sqrt{d}$.
Substituting $v_{\rm s}^{2}=1/d$ into the first condition~\eqref{con1} in Theorem~\ref{prop:con}, we obtain $\Phi|_{p}=\bm\sigma(\bar{\bm\xi},\bar{\bm\xi})|_{p}=0$.

($\Longleftarrow$)
Assume that $\bm\sigma(\bar{\bm\xi},\bar{\bm\xi})|_{p}=0$.
Combining $\bm\sigma(\bar{\bm\xi},\bar{\bm\xi})|_{p}=0$ with the first condition~\eqref{con1} in Theorem~\ref{prop:con}, we obtain
\begin{equation}
\left(v_{\rm s}^{2}-\frac{1}{d}\right)\Theta~\Big|_{p}=0.
\end{equation}
We have assumed $\Theta|_{p}\neq0$, and therefore the flow must satisfy $v_{\rm s}|_{p}=1/\sqrt{d}$.
\end{proof}

\begin{thm}
Consider the radiation fluid. 
The necessary conditions that the steady radiation fluid flow must satisfy at a sonic point $p$ reduce to
\begin{subequations}
\begin{eqnarray}
\bm\sigma(\bar{\bm\xi},\bar{\bm\xi})|_{p}&=&0,
\label{con1:rad}\\
{\cal L}_{\bar{\bm\eta}}\left[\bm\sigma(\bar{\bm\xi},\bar{\bm\xi})\right]|_{p}&\ge&0.
\label{con2:rad}
\end{eqnarray}
\end{subequations}
\label{prop:2}
\end{thm}
\begin{proof}
The speed of sound for the radiation fluid is given by $v_{\rm s}=1/\sqrt{d}$ from the equation of state.
Therefore, $\Phi$ for the radiation fluid case is written as $\Phi=\bm\sigma(\bar{\bm\xi},\bar{\bm\xi})$ at every point in $U$, and furthermore the right-hand side of the inequality~\eqref{con2} vanishes from the constancy of the speed of sound for the radiation fluid.
\end{proof}
\noindent
We note that, in Theorem\ref{prop:2}, the conditions are imposed on the time-time component of the shear tensor for the congruence of the streamlines 
in the radiation fluid case.

We further deform the second condition~\eqref{con2:rad} for the radiation fluid in preparation for checking a model of the sonic point/photon surface correspondence in Sec.~\ref{III}.
We can rewrite the inequality~\eqref{con2:rad} to
\begin{equation}
A_{1}+A_{2}+A_{3}\ge0,
\label{ineq}
\end{equation}
where 
\begin{subequations}
\begin{eqnarray}
A_{1}&:=&\left\{(\bm\sigma\cdot\bm\sigma-\bm\omega\cdot\bm\omega-2\bm\sigma\cdot\bm\omega)(\bar{\bm\xi},\bar{\bm\xi})-\frac{1}{d}\left(\bm\sigma:\bm\sigma+\bm\omega:\bm\omega\right)\right\}\Big|_{p},\label{A1}\\
A_{2}&:=&\left\{(\nabla\otimes\bm a^{\flat})(\bar{\bm\xi},\bar{\bm\xi})+\frac{1}{d}~\nabla\cdot\bm a\right\}\Big|_{p},\label{A2}\\
A_{3}&:=&\left\{-R(\bar{\bm\xi},\bar{\bm\eta},\bar{\bm\xi},\bar{\bm\eta})-\frac{1}{d}~{\rm Ric} (\bar{\bm\eta},\bar{\bm\eta})\right\}\Big|_{p}.\label{A3}
\end{eqnarray}
\end{subequations}
The colon ``~$:$~'' denotes double dot product defined by $\bm X:\bm Y={\rm tr}_{1,4}{\rm tr}_{2,3}(\bm X\otimes \bm Y)$ for tensor fields $\bm X, \bm Y$ of rank 2, while the dot product is given by $\bm X\cdot \bm Y={\rm tr}_{2,3}(\bm X\otimes\bm Y)$.
$R$ is the Riemann curvature tensor, and $\rm Ric$ is the Ricci tensor.
Readers may refer to Appendix~\ref{app:ineq} for the details of the
derivation of the inequality~\eqref{ineq}.

\section{Sonic point and photon surface}
\label{III}

In Sec.~\ref{II}, we obtained the necessary conditions~\eqref{con1} and
\eqref{con2} that the steady perfect fluid flow 
must satisfy at a sonic point.
In this section, we interpret these conditions to the conditions imposed on 
a certain timelike hypersurface crossing the sonic point 
which is a relativistic generalization of the section of the de Laval
nozzle (\ref{III1}).
The timelike hypersurface will reproduce the 
throat of the nozzle in the de Laval nozzle model, and it will be an unstable or marginally stable photon surface 
in the known examples of the sonic point/photon surface correspondence~\cite{Koga:2016jjq,Koga:2018ybs,Koga:2019teu}.
Then we summarize the relation between the conditions for sonic point and photon surface in Sec.~\ref{III2}. 

\subsection{proper section of congruence}
\label{III1}

We define the proper section of the congruence of streamlines around a streamline as follows:
\begin{defn}[proper section of a streamline congruence]
We define the proper section of a streamline congruence associated with a streamline 
as the codimension-1 foliation satisfying the following property: 
for every point $p$ on the streamline, there exists a timelike leaf $S\ni p$ of the foliation 
with the induced metric $\bm h$ in a manner such that every geodesic 
$\beta$ of $(S,\bm h)$ from $p$ 
satisfies
\begin{subequations}
\begin{eqnarray}
\bar{\bm\eta}\cdot\dot\beta~|_{p}&=&0,\label{perp:S}\\
\left\{\nabla_{\dot{\beta}}\dot{\beta}+\bm B(\dot{\beta},\dot{\beta})\bar{\bm\eta}\right\}\Big|_{p}&=&0,
\label{pro}
\end{eqnarray}
\end{subequations}
where $\dot\beta$ denotes the tangent vector of $\beta$.
\label{def:section}
\end{defn}
\noindent
Definition~\ref{def:section} stipulates that we always find the proper section $S$ for every point $p$ on a given streamline by emitting $\beta$ in all the directions orthogonal to $\bar{\bm\eta}|_{p}$, where the tangent bundle $TS$ is spanned by $\dot\beta$ and $d-1$ linearly independent Jacobi fields for the congruence of $\beta$.
We remark that the definition of the proper section presented above stipulates that we can find such a foliation even in the presence of the vorticity $\bm\omega$ on the streamline.
Note that, in all known examples of the sonic point/photon surface correspondence~\cite{Koga:2016jjq,Koga:2018ybs,Koga:2019teu}, 
not only the shear but also the vorticity vanishes at the sonic point, 
and $A_{1}$ in Eq.~\eqref{ineq} vanishes. 
We will see the details in Sec.~\ref{IIIc}.

We also have an alternative definition of 
a proper section to Definition~\ref{def:section} as in the following proposition:
\begin{prop}
A timelike hypersurface $S\ni p$ is the proper section for the point $p$ if and only if $S$ satisfies 
\begin{subequations}
\begin{eqnarray}
\left(\bar{\bm m}-\bar{\bm\eta}\right)|_{p}&=&0,\label{eta:m0}\\
\left(\bm\chi-\bm B^{{\rm (S)}\perp}\right)|_{p}&=&0,\label{hB}
\end{eqnarray}
\end{subequations}
where $\bar{\bm m}$ is the unit vector field normal to $S$, $\bm\chi$ is the second fundamental form for $S$, and $\bm B^{{\rm (S)}\perp}$ denotes the symmetric part of $\bm B^{\perp}$:
\begin{equation}
\bm B^{{\rm (S)}\perp}|_{p}:=\left(\frac{\Theta}{d}\,\bm g^{\perp}+\bm\sigma\right)\Big|_{p}.
\label{SymB}
\end{equation}
\label{prop:S}
\end{prop}
\begin{proof}
($\Longleftarrow$)
Assume that $S$ satisfies Eqs.~\eqref{eta:m0} and \eqref{hB}.
Consider an arbitrary geodesic $\beta$ of $(S,\bm h)$ crossing $p$.
Because $\beta\subset S$, $\dot{\beta}$ immediately obeys $\bar{\bm m}\cdot\dot\beta=0$, and therefore $\dot{\beta}$ satisfies Eq.~\eqref{perp:S} from Eq.~\eqref{eta:m0} for any $\beta$.
Then $\dot\beta$ obeys the geodesic equation $\nabla_{\dot\beta}\dot\beta+\bm\chi(\dot\beta,\dot\beta)\bar{\bm m}=0$.
From Eqs.~\eqref{eta:m0} and \eqref{hB}, replacing $\bm \chi$ and $\bar{\bm m}$ in the second term of the geodesic equation with $\bm B^{{\rm (S)}\perp}$ and $\bar{\bm\eta}$, respectively, we obtain Eq.~\eqref{pro}.

($\Longrightarrow$) 
Assume that $S$ satisfies Eqs.~\eqref{perp:S} and \eqref{pro}.
Consider an arbitrary geodesic $\beta$ of $(S,\bm h)$ crossing $p$.
Because $\beta\subset S$, $\dot{\beta}$ immediately obeys $\bar{\bm m}\cdot\dot\beta=0$.
Combining $\bar{\bm m}\cdot\dot\beta=0$ with Eq.~\eqref{perp:S}, we obtain $(\bar{\bm m}-\bar{\bm\eta})|_{p}=c\,\bar{\bm m}|_{p}$, where a constant $c$ is determined to either of $c=0$ or $c=2$ from the normalization conditions on $\bar{\bm m}$ and $\bar{\bm\eta}$, which depends on the orientation of $S$ by $\bar{\bm m}$.
Choosing 
the appropriate orientation $c=0$, 
we obtain Eq.~\eqref{eta:m0}.
Then, from Eq.~\eqref{pro} and the geodesic equation 
on the hypersurface $S$ for $\beta$,
$\nabla_{\dot{\beta}}\dot{\beta}+\bm \chi(\dot{\beta},\dot{\beta})\bar{\bm m}=0$,
we have 
$\left(\bm\chi-\bm B^{(\rm S)}\right)(\dot{\beta},\dot{\beta})|_p=0$ 
for all the geodesics $\beta\ni p$ of $(S,\bm h)$ 
where we also used Eq.~\eqref{eta:m0}.
Here, for any geodesics $\beta_{1},\beta_{2}\ni p$ of $(S,\bm h)$, there is another geodesic $\beta_{3}\ni p$ of $(S,\bm h)$ such that $\dot\beta_{1}|_{p}+\dot\beta_{2}|_{p}$ and $\dot\beta_{3}|_{p}$ are linearly dependent.
Therefore, we arrive at $\left(\bm\chi-\bm B^{(\rm S)}\right)(\dot{\beta}_{1}+\dot{\beta}_{2},\dot{\beta}_{1}+\dot{\beta}_{2})|_p=0$.
Applying $\left(\bm\chi-\bm B^{(\rm S)}\right)(\dot{\beta}_{1},\dot{\beta}_{1})|_p=\left(\bm\chi-\bm B^{(\rm S)}\right)(\dot{\beta}_{2},\dot{\beta}_{2})|_p=0$ gives $\left(\bm\chi-\bm B^{(\rm S)}\right)(\dot{\beta}_{1},\dot{\beta}_{2})|_p=0$.
Then we obtain Eq.~\eqref{hB}.
\end{proof}
%
%
\noindent
It should be emphasized that, through Proposition~\ref{prop:S}, 
the shear tensor
$\bm \sigma$ in Theorem~\ref{prop:2} can be regarded as the traceless part 
${\bm
\sigma}_{\bm \chi}$
of the 
second fundamental form for the proper section of the fluid flow.
That is, letting $H$ be the mean curvature for $S$, and 
${\bm
\sigma}_{\bm \chi}$ be the trace-free part of $\chi$, we find 
${\bm \chi} = H{\bm h} + {\bm \sigma}_{\bm \chi}$, and 
\begin{subequations}
\begin{eqnarray}
\frac{1}{d}\Theta\Big|_{p} &=& H|_{p}, \\
{\bm \sigma}|_{p}&=&{\bm \sigma}_{\bm \chi}|_{p}.
\end{eqnarray}
\end{subequations}
Therefore, 
the conditions for the sonic point can be interpreted as conditions for 
the associated proper section, and we can compare them with the conditions for 
the photon surface. 
However, we note that 
$\bar{\bm \eta}$ is identical to $\bar{\bm m}$ only at the point $p$ 
and not hyper-surface normal in general. 
Thus 
the Lie derivative of $\bm \sigma$ along $\bar{\bm \eta}$ 
included in Eq.~\eqref{con2:rad} should be carefully evaluated.

\subsection{Sonic point on photon surface}
\label{III2}

Let us summarize the conditions for the sonic point and photon surface, and compare them with each other. 
For completeness, here we describe the definition of the stability of a photon surface: 
\begin{defn}[Unstable or marginally stable photon surface~\cite{Claudel:2000yi,Perlick:2005jn,Koga:2019uqd}]
A photon surface of a spacetime $(M,\bm g)$ of arbitrary dimension is an immersed, nowhere-spacelike hypersurface $S$ 
such that, for every point $p \in S$ and every null vector $\bm k \in T_{p}S$, there exists a null geodesic $\gamma : (-\epsilon, \epsilon) \to M$ 
with $\gamma(0) =\bm k$ and $|\gamma| \subset S$.

Let $R$ be the Riemann curvature tensor associated with the Levi-Civita connection on $(M,\bm g)$ and let $\bar{\bm m}\in N_{p}S$ be the unit vector normal to $S$.
A photon surface is said to be unstable {\rm (}marginally stable{\rm )} 
if $R(\bar{\bm m},\bm k,\bar{\bm m},\bm k)<0$~{\rm (}$R(\bar{\bm m},\bm k,\bar{\bm m},\bm k)=0${\rm )} for  every point $p \in S$ and every null vector $\bm k \in T_{p}S$.
\label{def:1}
\end{defn}
The following theorem provides an alternative equivalent condition for the stability: 
\begin{thm}[\cite{Claudel:2000yi,Perlick:2005jn,Koga:2019uqd}]
Let $\{S_{r}\}_{r\in(-\epsilon,\epsilon)}$ be a timelike Gaussian normal foliation of $(M,\bm g)$ such that $S=S_{0}$ is a photon surface. 
Let $\bm\sigma_{\bm\chi}$ be the trace-free part of the second fundamental form for $\{S_{r}\}_{r\in(-\epsilon,\epsilon)}$. 
A photon surface is unstable {\rm (}marginally stable{\rm )} if and only if $(\nabla_{\bar{\bm m}}\bm\sigma_{\bm\chi})(\bm k,\bm k)>0$~{\rm (}$(\nabla_{\bar{\bm m}}\bm\sigma_{\bm\chi}{\rm )}(\bm k,\bm k)=0${\rm )} 
for  every point $p \in S$ and every null vector $\bm k \in T_{p}S$. 
\label{thm:stbsig}
\end{thm}

Combining Theorems~\ref{def:2} and \ref{thm:stbsig}, we obtain the conditions 
for an unstable or marginally stable photon surface as follows:
\begin{subequations}
\begin{eqnarray}
\bm \sigma_{\bm \chi}|_p=0, 
\label{eq:pscon1}
\\
(\nabla_{\bar{\bm m}} \bm \sigma_{\bm \chi}) (\bm k,\bm k)|_p\geq 0, 
\label{eq:pscon2}
\end{eqnarray}
\end{subequations}
where we note that $\bm \sigma|_p=\bm \sigma_{\bm \chi}|_p$ but $\bm \sigma\neq\bm \sigma_{\bm \chi}$ 
at the other points on the stream line in general. 
These conditions should be compared with the conditions \eqref{con1:rad} and \eqref{con2:rad}. 
It can be easily found that \eqref{con1:rad} guarantees only a part of the condition \eqref{eq:pscon1}, 
and the inequality \eqref{con2:rad} is not equivalent to \eqref{eq:pscon2}. 
Therefore, 
we conclude that, in all known examples for the sonic point/photon surface correspondence, 
additional assumptions, such as spatial symmetry and specific fluid configuration, are essential 
for the realization of the correspondence. 
More concretely, regarding the equality condition on the shear tensor,
$\bm\sigma(\bar{\bm\xi},\bar{\bm\xi})|_{p}=0$ should imply $\bm\sigma_{\bm\chi}=0$ together with 
the additional assumptions 
for realization of the sonic point/photon surface correspondence. 
That is, 
Eq.~\eqref{con1:rad} is a necessary condition for the proper section being a photon surface.
From this perspective, we may understand that 
the sonic point/photon surface correspondence would require 
spatial $(d-1)$-dimensional maximal symmetry to the spacetime: $G=SO(d)$, $E(d-1)$ or $SO(d-1,1)$.

Let us consider the general equation of state for the perfect fluid. 
Then 
Proposition~\ref{prop:5} leads to the following proposition 
\begin{prop}
Suppose that $p$ is a sonic point.
If the proper section $S\ni p$ is a photon surface whose mean curvature does not vanish at $p$, the speed of sound at $p$ must be given by $1/\sqrt{d}$.
\label{prop:6}
\end{prop}
\begin{proof}
Denoting the mean curvature of the second fundamental form $\bm \chi$ on the photon surface by $H$, we find 
$\frac{1}{d}\Theta|_p=H|_p\neq 0$ and $\bm \sigma(\bar{\bm \xi},\bar{\bm \xi})|_p=\bm \sigma_{\bm \chi}(\bar{\bm \xi},\bar{\bm \xi})|=0$. 
Therefore, 
Proposition~\ref{prop:5} applies, i.e. we get $v_{\rm s}|_{p}=1/\sqrt{d}$.
\end{proof}

So far, 
we have considered the conditions for a given sonic point to be on the photon surface.
On the other hand, we may consider the following question: is it possible to find a solution of 
the steady radiation fluid flow normal to a given photon surface 
with a sonic point on the photon surface?  
Practically, 
through the equations of motion \eqref{1st} - \eqref{eos}, we can 
specify the configuration of  
the steady radiation fluid flow 
in the vicinity of the timelike hypersurface $S$ 
by a set of functional forms of 
$P$, $s$, $v$ and the vector field $\bar{\bm\eta}$ on $S$.  
In other words, we regard the solutions 
are equivalent to each other 
if they 
satisfy  
the same boundary condition on $S$.
Therefore, in order to answer our question, 
we have to consider the conditions \eqref{con1:rad} and \eqref{con2:rad} 
on the photon surface $S$ as a boundary. 
Let us consider the boundary condition satisfying the following equations:
\begin{subequations}
\begin{eqnarray}
\bm\sigma|_{p}=\bm\omega|_{p}&=&0,\label{eta:m}\\
\dd_{S}P\,|_{p}=\dd_{S}s\,|_{p}=\dd_{S}v\,|_{p}&=&0,
\label{dSPsv}
\end{eqnarray}
\end{subequations}
where $\dd_{S}$ denotes the exterior derivative on $S$. 
Setting $v|_p=v_{\rm s}=1/\sqrt{d}$, the condition \eqref{con1:rad} can be trivially satisfied. 
We also find that we can make the condition~\eqref{con2:rad} satisfied by 
appropriately setting the functional form of $P$ on the photon surface $S$ 
through the following Proposition \ref{prop:7}. 
\begin{prop}
Consider a steady flow of the radiation fluid. 
Suppose that $p$  
 is the sonic point, and that $p$  satisfies Eqs.~\eqref{eta:m} and \eqref{dSPsv}. 
Then $S$ must satisfy the inequality
\begin{equation}
A_{1}+A_{2}+A_{3}\ge0, 
\label{con2:S}
\end{equation}
where
\begin{subequations}
\begin{eqnarray}
A_{1}&=&0,\label{A1:prop}\\
A_{2}&=&-\left\{\Delta_{S}\left(\ln\sqrt{|\bm\xi\cdot\bm\xi|}\right)+\frac{2}{d}~{\rm Ric}(\bar{\bm m},\bar{\bm\xi})+\left(1-\frac{1}{d}\right)(nh)^{-1}\Delta_{S}P\right\}\Big|_{p},\label{A2:prop}\\
A_{3}&=&-\left\{R(\bar{\bm\xi},\bar{\bm m},\bar{\bm\xi},\bar{\bm m})+\frac{1}{d}~{\rm Ric} (\bar{\bm m},\bar{\bm m})\right\}\Big|_{p},\label{A3:prop}
\end{eqnarray}
\end{subequations}
and $\Delta_{S}$ denotes the Laplace-Beltrami operator on $(S,\bm h)$ with $\bm h$ being the induced metric on $S$.
\label{prop:7}
\end{prop}
\begin{proof}
The expressions~\eqref{A1:prop} and \eqref{A3:prop} of $A_{1}$ and $A_{3}$ immediately follow from Eq.~\eqref{eta:m} and \eqref{eta:m0}, respectively.
We give the derivation of the expression~\eqref{A2:prop} of $A_{2}$ in Appendix~\ref{app:ineq2}.
\end{proof}
\noindent
From the equation of state~\eqref{eos} for the radiation fluid, $n$ and $h$ for the radiation fluid are related with $P$ by
\begin{equation}
nh=\left(\frac{1}{d}\right)^{-1}\left(1+\frac{1}{d}\right)P.
\label{nh}
\end{equation}
Substituting Eq.~\eqref{nh} into Eq.~\eqref{A2:prop}, we can rewrite the inequality~\eqref{con2:S} to the 
following inequality: 
\begin{equation}
P^{-1}\Delta_{S}P\,|_{p}\le-\frac{\left(1+\frac{1}{d}\right)}{\left(\frac{1}{d}\right)\left(1-\frac{1}{d}\right)}\left\{\Delta_{S}\left(\ln\sqrt{|\bm\xi\cdot\bm\xi|}\right)+\frac{2}{d}~{\rm Ric}(\bar{\bm m},\bar{\bm\xi})+R(\bar{\bm\xi},\bar{\bm m},\bar{\bm\xi},\bar{\bm m})+\frac{1}{d}~{\rm Ric} (\bar{\bm m},\bar{\bm m})\right\}\Big|_{p}.  
\label{ineq:P}
\end{equation}
Apparently, choosing the functional form of $P$ such that the left-hand side $P^{-1}\Delta_{S}P|_{p}$ is less than 
the right-hand side of the inequality~\eqref{ineq:P}, 
we can make the condition \eqref{con2:rad} satisfied.

\section{Sonic point/photon surface correspondence}
\label{IIIc}

In this section, we reanalyze known examples of 
the sonic point/photon surface correspondence from the view point of the shear tensor 
for the congruence of the streamlines and the trace-free part of the second fundamental form 
based on Refs.~\cite{Koga:2016jjq,Koga:2018ybs,Koga:2019teu}.

Following Ref.~\cite{Koga:2019teu}, we consider a $(d+1)$-dimensional static spacetime with spatially spherical, planar or hyperbolic symmetry $G$ whose metric is given by 
\begin{equation}
\bm g=g_{tt}(r)\dd t^{2}+g_{rr}(r)\dd r^{2}+\bm\gamma, 
\label{met}
\end{equation}
where smooth functions $g_{tt}(r)$ and $g_{rr}(r)$ satisfy $g_{tt}(r)<0$ and $g_{rr}(r)>0$, respectively.
The induced metric $\bm\gamma$ of a spacelike submanifold  $\{r,t=const.\}$ is given by
\begin{equation}
\bm \gamma=r^{2}\left\{\dd\theta^{2}+s(\theta)^{2}\dd\Omega_{(d-2)}^{2}\right\},~~~
s(\theta)=\begin{cases}\sin\theta&(~G=SO(d)~)\\
\theta&(~G=E(d-1)~)\\
\sinh\theta &(~G=SO(d-1,1)~)
\end{cases},
\label{max}
\end{equation}
where $\dd\Omega_{(d-2)}^{2}$ is the metric of the unit $(d-2)$-sphere.
By introducing the polar coordinates $(\phi_{1},\cdots,\phi_{d-2})$, we write $\dd\Omega_{(d-2)}^{2}$ as
\begin{equation}
\dd\Omega_{(d-2)}^{2}=\sum_{i=1}^{d-2}\left(\prod_{j=1}^{i}\sin^{2}\phi_{j}\right)\frac{\dd\phi^{2}_{i}}{\sin^{2}\phi_{i}},
\end{equation}
and we write the $(d-2)$-th coordinate $\phi_{d-2}$ by $\phi_{d-2}=:\phi$ in order to distinguish $\phi_{d-2}$ as the azimuthal coordinate.
The corresponding basis $\bm\partial_{\phi}$ to $\phi$ is a spacelike Killing vector field on the spacetime.
We also denote $\theta$ by $\theta=\phi_{0}$ for convenience.

Following Ref.~\cite{Koga:2018ybs}, we investigate the steady perfect fluid flow 
with Killing observers on the equatorial plane 
$\{ \phi_{0}=\cdots=\phi_{d-3}=\pi/2\}$.
We impose the following several assumptions on the flow as in Ref.~\cite{Koga:2018ybs}.
Letting 
$\rho^{2}:=\phi_{0}^{2}+\cdots+\phi_{d-3}^{2}$, 
we define the equivalence relation ``\,$\sim$\,'' around the equatorial plane to the first order by $f\sim g:\Longleftrightarrow f=g+{\cal O}(\rho^{2})$ for functions $f$ and $g$, and $\bm X\sim \bm Y:\Longleftrightarrow X^{\mu}=Y^{\mu}+{\cal O}(\rho^{2})$ for vector fields $\bm X$ and $\bm Y$.
We suppose that the perfect fluid flow admits the translational symmetries associated with $t$ and $\phi$, and reversal symmetries associated with 
$(\phi_{0},\cdots,\phi_{d-3})$, 
i.e.
\begin{subequations}
\begin{eqnarray}
P&\sim&P(r),\label{P:c}\\
s&\sim&s(r),\label{s:c}\\
\bm u&\sim&u^{t}(r)\bm\partial_{t}+u^{r}(r)\bm\partial_{r}+u^{\phi}(r)\bm\partial_{\phi}+\sum_{i=0}^{d-3}\phi_{i}\,u^{\phi_{i}}(r)\bm\partial_{\phi_{i}},
\end{eqnarray}
\end{subequations}
and we also suppose that the flow also satisfies the following additional conditions:
\begin{subequations}
\begin{eqnarray}
u^r(r)&\neq&0,\\
u^{\phi_{i}}(r)&=&0 \quad (i=0,\cdots, d-3).
\label{add}
\end{eqnarray}
\end{subequations}
The former condition leads to
\begin{equation}
s\sim const.\nonumber
\end{equation}
from Eq.~\eqref{ise}.
In Ref.~\cite{Koga:2018ybs}, the additional condition~\eqref{add} was imposed as the condition of uniform matter distribution of the disk. 

So far, we did not specify the observer $(d+1)$-velocity $\bm u_{\rm obs}$ nor the Killing vector field $\bm \xi$. 
Here, we reemphasize that $\bm u_{\rm obs}$ and $\bm \xi$ are not identical to each other in general, but 
satisfy $\bm u_{\rm obs}|_p=\bar{\bm \xi}|_p$ at the point $p$ under consideration. 
In the present setting, let us consider a set of co-rotating observers. 
The vector field associated with the co-rotating observers takes the following form on the equatorial plane:
\begin{equation}
\bm u_{\rm obs}\sim\gamma(r)\{\bm\partial_{t}+\omega(r)\bm\partial_{\phi}\}, 
\label{uo}
\end{equation}
where $\omega(r):=u^{\phi}(r)/u^{t}(r)$ 
denotes the angular velocity of the observers, and $\gamma(r)$ is the normalization factor which ensures $\bm u_{\rm obs}\cdot\bm u_{\rm obs}=-1$.
We note that the model of the rotating flow in the $G=SO(d)$ case~\cite{Koga:2018ybs} has not been extended to the $G=E(d-1)$ and $G=SO(d-1,1)$ cases although one can easily predict the quite similar results of the sonic point/photon surface correspondence to Ref.~\cite{Koga:2018ybs}. 
In specifying the speed of the flow, the orthogonal decomposition of $\bm u$ with respect to $\bm u_{\rm obs}$ is performed, and the spacelike unit vector field $\bar{\bm w}$ orthogonal to $\bm u_{\rm obs}$ defined from the orthogonal decomposition obeys the following equivalence relation:
\begin{equation}
\bar{\bm w}\sim g_{rr}^{-1/2}\bm\partial_{r}.
\label{w}
\end{equation}
Therefore, we can regard $\bar{\bm w}$ as the 
unit vector field normal to $r=const.$ hypersurfaces 
in the vicinity of the equatorial plane.  

For each point $p\in\{\phi_{0}=\cdots=\phi_{d-3}=\pi/2\}$ on the equatorial plane, we employ the Killing vector field 
\begin{equation}
\bm\xi:=\bm\partial_{t}+\omega(r_{p})\bm\partial_{\phi}, 
\label{xi:check}
\end{equation}
where $r_{p}$ denotes the radius 
at $p$.
We get the spacelike vector field $\bar{\bm\eta}$ that is orthogonal to $\bm\xi$ and obeys
\begin{equation}
\bar{\bm\eta}\sim\bar{\eta}\,^{t}(r)\bm\partial_{t}+\bar{\eta}\,^{r}(r)\bm\partial_{r}+\bar{\eta}\,^{\phi}(r)\bm\partial_{\phi}
\label{ex_eta}
\end{equation}
on a neighborhood $U$ of $p$. 
By construction, $\bm u_{\rm obs}$ coincides with $\bar{\bm \xi}$ at $p$, and simultaneously $\bar{\bm w}$ coincides with $\bar{\bm\eta}$ at $p$.
Comparing Eqs.~\eqref{xi:check} and \eqref{ex_eta} at $p$, we obtain $\bar{\eta}\,^{t}(r_{p})=\bar{\eta}\,^{\phi}(r_{p})=0$, and therefore
\begin{equation}
\left(\bar{\bm w}-\bar{\bm\eta}\right)|_{S}\sim0,\label{weta}
\end{equation}
where $S=\{r=r_{p}\}$ is the timelike hypersurface of constant radius
including $p$. 
In other words, we can regard $\bar {\bm \eta}$ as the unit vector normal to $S$, 
and identify it with $\bar {\bm w}$ in the vicinity of $p$. 
Then $\bm B^{(S)\perp}$ can be also identified with $\bm \chi$ at $p$:
\begin{equation}
\left(\bm \chi-\bm B^{({\rm S})\perp}\right)|_p = 0, 
\end{equation}
where $\bm \chi$ is the second fundamental form of $S$.
As a consequence, on each point on the equatorial plane, the timelike hypersurface $S$ of constant $r$ including the point is the proper section, defined in Sec.~\ref{III1}, according to Proposition~\ref{prop:S}.

Hereafter, we investigate the necessary conditions that the steady perfect fluid flow must satisfy at the sonic point for this model.
First of all, we show the following lemma for the 
proper section $S$ which plays an important role in the sonic point/photon surface correspondence:
\begin{lem}
Let $S=\{r=r_{p}\}$ be a timelike hypersurface of constant $r$ in the spacetime of the metric~\eqref{met} with \eqref{max}, and let  $\bm\sigma_{\bm\chi}$ be the trace-free part of the second fundamental form $\bm\chi$ for $S$.
For any tangent vectors $\bm X,\bm Y\in T_{p}S$ such that $X^{\mu}, Y^{\mu}\neq0$, the following statements for $\bm\sigma_{\bm\chi}$ hold:
\begin{itemize}
\item
$\bm\sigma_{\bm\chi}(\bm X,\bm X)=0$ if and only if $\bm\sigma_{\bm\chi}(\bm Y,\bm Y)=0$,
\item
$\bm\sigma_{\bm\chi}(\bm X,\bm X)>0$ if and only if $\bm\sigma_{\bm\chi}(\bm Y,\bm Y)>0$,
\item
$\bm\sigma_{\bm\chi}(\bm X,\bm X)<0$ if and only if $\bm\sigma_{\bm\chi}(\bm Y,\bm Y)<0$.
\end{itemize}
\label{lem:2}
\end{lem}
\begin{proof}
In the following, we substitute $r=r_{p}$ into all the functions of $r$ without writing explicitly for convenience.
The components of $\bm \chi$ are given as
\begin{subequations}
\begin{eqnarray}
\chi_{ij}&=&\Lambda g_{ij},\\
\chi_{it}&=&0,\\
\chi_{tt}&=&\frac{1}{2}(g_{rr})^{-1/2}\frac{\rm d}{{\rm d}r}g_{tt},
\end{eqnarray}
\end{subequations}
where $i,j\in\{\theta,\phi_{1},\cdots,\phi_{d-3},\phi\}$, and $\Lambda$ is defined by $\Lambda^{-1}:=r\sqrt{g_{rr}}$, and the components of $\bm\sigma$ are given as
\begin{subequations}
\begin{eqnarray}
(\sigma_{\bm\chi})_{ij}&=&\left(\Lambda-H\right)g_{ij},\\
(\sigma_{\bm\chi})_{it}&=&0,\\
(\sigma_{\bm\chi})_{tt}&=&\chi_{tt}-Hg_{tt},\label{chi:eps1}
\end{eqnarray}
\end{subequations}
where $H$ denotes the mean curvature for $S$.
Here, the definition of $H$ gives another expression of $\chi_{tt}$ as follows:
\begin{equation}
Hd=\Lambda(d-1)+g^{tt}\chi_{tt}.
\label{chi:eps2}
\end{equation}
Applying Eq.~\eqref{chi:eps2} to Eq.~\eqref{chi:eps1} gives
\begin{equation}
(\sigma_{\bm\chi})_{tt}=\left(\Lambda-H\right)(d-1)|g_{tt}|.
\end{equation}
The $(\bm X,\bm X)$ component 
of $\bm\sigma_{\bm \chi}$ is given as
\begin{equation}
\bm\sigma_{\bm\chi}(\bm X,\bm X)=\left(\Lambda-H\right)\left\{(d-1)|g_{tt}|\left(X^{t}\right)^{2}+g_{ij}X^{i}X^{j}\right\}.
\end{equation}
The curly bracket is positive as long as $X^{\mu}\neq0$.
Therefore, the coefficient $\Lambda-H$ determines the sign independently of $\bm X$.
\end{proof}
%
\noindent
The first statement of the equality ``$\,=\,$'' in Lemma~\ref{lem:2} implies that $\bm\sigma(\bar{\bm\xi},\bar{\bm\xi})|_{p}=0$ holds for the congruence of streamlines if and only if $\bm\sigma_{\bm\chi}=0$ holds for the 
proper section 
$\{r=r_{p}\}$.
The mean curvature $H$ takes the positive value $H=\Lambda(r_{p})>0$, and we arrive at this result: the speed of sound at the sonic point $p$ is given by $1/\sqrt{d}$ if and only if the proper section 
$\{r=r_{p}\}$ is a photon surface
, where the
`if' part comes from Proposition~\ref{prop:6}, and the `only if' part comes from Proposition~\ref{prop:5} and Lemma~\ref{lem:2}.
We also find that the sonic point for the radiation fluid must be on a photon surface.

Hereafter, we consider the radiation fluid. 
The flow fulfills both of the assumptions~\eqref{eta:m} and \eqref{dSPsv}, 
i.e. Proposition~\ref{prop:7} applies to this model.
We investigate the inequality~\eqref{con2:S} that the proper section must satisfy for this model, 
and show that the inequality~\eqref{con2:S} implies the inequality \eqref{eq:pscon2} in the present setting. 
We first deform $A_{2}$.
Here, we start with the expression~\eqref{A2} of $A_{2}$ in terms of $\bm a$ instead of Eq.~\eqref{A2:prop}.
The vector field $\bm a$ is orthogonal to $\bar{\bm\eta}$ from Eq.~\eqref{a eta}, and $\bar{\bm\eta}$ is regarded as 
the vector normal to $S=\{r=r_{p}\}$
in the vicinity of the equatorial plane. 
Therefore, we obtain $a^{r}(r_{p})\sim0$, and we find the $r$ component of $\bm a$ vanishes at $p$.
Introducing the determinant $g$ of the metric $\bm g$ in the coordinate system $(t,r,\theta,\phi_{1},\cdots,\phi_{d-3},\phi)$, we can deform the expression~\eqref{A2} of $A_{2}$ to
\begin{align}
A_{2}&=\Big\{a^{\mu}\partial_{\mu}\left(\ln\sqrt{|\xi^{\nu}\xi_{\nu}|}\right)+|g|^{-1/2}\partial_{\mu}\left(|g|^{1/2}a^{\mu}\right)\Big\}\Big|_{p}\notag\\
&=\partial_{r}a^{r}\,|_{p}\notag\\
&={\cal L}_{\bar{\bm\eta}}\bm a\cdot\bar{\bm\eta}\,|_{p}.
\label{A2:ineq}
\end{align}
Here, we consider the identity ${\cal L}_{\bar{\bm\eta}}\left[\nabla_{\bar{\bm\eta}}(\bar{\bm\eta}\cdot\bar{\bm\eta})\right]=0$ which is deformed to
\begin{equation}
{\cal L}_{\bar{\bm\eta}}\bm a\cdot\bar{\bm\eta}+\bm a\cdot\bm a=0.
\label{A2:identity}
\end{equation}
The first term of Eq.~\eqref{A2:identity} is identical to $A_{2}$.
From Eq.~\eqref{a xi}, $\bm a$ is spacelike, and therefore we have the inequality
\begin{equation}
A_{2}\le0.
\label{A2:ineq2}
\end{equation}
We also have the expression~\eqref{A2:prop} of $A_{2}$ in terms of the pressure $P$.
Given the metric~\eqref{met} with \eqref{max} of the spacetime, we get
\begin{subequations}
\begin{eqnarray}
\Delta_{S}\left(\ln\sqrt{|\bm\xi\cdot\bm\xi|}\right)&=&0,\\
{\rm Ric}(\bar{\bm m},\bar{\bm\xi})&=&0,
\end{eqnarray}
\end{subequations}
where $\bar{\bm m}$ is the unit vector normal to $S$.
Eq.~\eqref{A2:prop} reduces to
\begin{equation}
A_{2}=-\frac{\left(\frac{1}{d}\right)\left(1-\frac{1}{d}\right)}{\left(1+\frac{1}{d}\right)}P^{-1}\Delta_{S}P~\Big|_{p}.
\end{equation}
The inequality~\eqref{A2:ineq2} leads to $\Delta_{S}P|_{p}\ge0$
\footnote{
\baselineskip 4.5mm
Remark that the thin disk model in Ref.~\cite{Koga:2018ybs} dealt with the equality case $\Delta_{S}P\,|_{p}=0$  by assuming uniform matter distribution in the angular directions.}. 

Applying the inequality~\eqref{A2:ineq2} to Proposition~\ref{prop:7} gives the inequality for $A_{3}$:
\begin{equation}
A_{3}\ge-A_{2}\ge0
\label{A3:ineq}
\end{equation}
as a necessary condition.  
$A_{3}$ is originally a part of the left-hand side of the inequality~\eqref{con1:rad}, and we performed the 
decomposition of the inequality~\eqref{con1:rad} in Appendix~\ref{app:ineq}.
In this model, $A_{1}$ vanishes, and therefore $A_{3}$ is the remaining part of the inequality for the congruence of the streamlines in addition to  
$A_{2}$, 
which contains all contributions from $\bm a$. 
In other words, $A_{3}$ 
does not depend on $\bm a$, and the value of $A_3$ is shared by the shear tensor $\bm \sigma_{\bm \chi}$ 
associated with the Gaussian normal foliation with the normal vector $\bar{\bm m}=\bar{\bm w}$ 
adopted in Theorem~\ref{thm:stbsig}. 
Therefore, hereafter, we consider $A_3$ as a variable associated with the Gaussian normal coordinate.  
We can rewrite $A_{3}$ as
\begin{equation}
A_{3}={\cal L}_{\bar{\bm w}}\left[\bm\sigma_{\bm \chi}(\bar{\bm\xi},\bar{\bm\xi})\right]\Big|_{p}. 
\label{A3:ineq3}
\end{equation}
One practically finds that Eq.~\eqref{A3:ineq3}
can be shown by performing the deformation of the right-hand side of Eq.~\eqref{A3:ineq3} together with Eq.~\eqref{w}, following Appendix~\ref{app:ineq}
\footnote{
\baselineskip 4.5mm
In Appendix~\ref{app:ineq}, we performed the decomposition ${\cal L}_{\bar{\bm\eta}}[\bm\sigma(\bar{\bm\xi},\bar{\bm\xi})]|_{p}=C_{1}+C_{2}+C_{3}=A_{1}+A_{2}+A_{3}$.
In the same manner, we can also perform the decomposition ${\cal L}_{\bar{\bm w}}[\bm\sigma_{\bm\chi}(\bar{\bm\xi},\bar{\bm\xi})]|_{p}=(C_{\bar{\bm w}})_{1}+(C_{\bar{\bm w}})_{2}+(C_{\bar{\bm w}})_{3}$ for the Gaussian normal foliation, where $(C_{\bar{\bm w}})_{1}:={\cal L}_{\bar{\bm w}}H|_{p}$, $(C_{\bar{\bm w}})_{2}:=(\nabla_{\bar{\bm w}}\bm\chi)(\bar{\bm\xi},\bar{\bm\xi})|_{p}$ and $(C_{\bar{\bm w}})_{3}:=2\bm\chi(\nabla_{\bar{\bm\xi}},\bar{\bm\xi})|_{p}$ are calculated to be
$(C_{\bar{\bm w}})_{1}=\left\{-H^{2}-\frac{1}{d}{\rm Ric}(\bar{\bm\eta},\bar{\bm\eta})\right\}\Big|_{p}$,
$(C_{\bar{\bm w}})_{2}=\Big\{H^{2}-R(\bar{\bm w},\bar{\bm\xi},\bar{\bm w},\bar{\bm\xi})\Big\}\Big|_{p}$ and
$(C_{\bar{\bm w}})_{3}=0$, respectively.
We used $\sigma_{\bm\chi}|_{p}=\bm\omega_{\bm\chi}|_{p}=\bm a_{\bm\chi}|_{p}=0$ from the assumptions in Sec.~\ref{IIIc}.
We have Eq.~\eqref{A3:ineq3}: ${\cal L}_{\bar{\bm w}}[\bm\sigma_{\bm\chi}(\bar{\bm\xi},\bar{\bm\xi})]|_{p}=\left\{-\frac{1}{d}{\rm Ric}(\bar{\bm\eta},\bar{\bm\eta})-R(\bar{\bm w},\bar{\bm\xi},\bar{\bm w},\bar{\bm\xi})\right\}|_{p}=A_{3}$.
}.
Here, we label the $r=const.$ hypersurfaces around the sonic point $p$ by a real number $\epsilon$ as $S_{\epsilon}:=\{r=r_{p}+\epsilon\}$, and we denote the trace-free part of the second fundamental form for 
$S_{\epsilon}$ by $\bm\sigma_{\epsilon}$.
We have $S_{0}=S$ and 
$\bm\sigma_{0}=\bm\sigma_{\bm\chi}|_p$ 
with this notation.
Eq.~\eqref{A3:ineq3} is further deformed to be
\begin{equation}
A_{3}=\lim_{\epsilon\to0}\frac{\bm\sigma_{\epsilon}(\bar{\bm\xi},\bar{\bm\xi})-\bm\sigma_{0}(\bar{\bm\xi},\bar{\bm\xi})}{\sqrt{g_{rr}(r_{p})}~\epsilon}
=\lim_{\epsilon\to0}\frac{\bm\sigma_{\epsilon}(\bar{\bm\xi},\bar{\bm\xi})}{\sqrt{g_{rr}(r_{p})}~\epsilon}, 
\label{A3:ineq4}
\end{equation}
where we have used $\bm\sigma_{0}(\bar{\bm\xi},\bar{\bm\xi})=0$. 
Recall that $A_{3}$ must satisfy the inequality~\eqref{A3:ineq}.
Applying the second statement of Lemma~\ref{lem:2} to the
inequality~\eqref{A3:ineq} with the expression~\eqref{A3:ineq4} of
$A_{3}$, we arrive at the following inequality including an arbitrary
null vector field 
$\bm k$ on $M$s:
\begin{equation}
{\cal L}_{\bar{\bm w}}\left[\bm\sigma_{\bm \chi}(\bm k,\bm k)\right]\Big|_{p}\ge0.
\label{inst}
\end{equation}
In the investigation of the stability of the photon surface, we consider the deviation of null geodesics whose initial tangent vectors are parallelly transported from each other.
Applying $\nabla_{\bar{\bm w}}\bm k|_{p}=0$ to the inequality~\eqref{inst} gives
\begin{equation}
(\nabla_{\bar{\bm w}}\bm\sigma_{\bm \chi})(\bm k,\bm k)\Big|_{p}\ge0,
\label{inst2}
\end{equation}
where we have supposed that $\bm k\cdot\bm k|_{p}=0$, and $\bm k|_{p}\in T_{p}S$.
The inequality~\eqref{inst2} is exactly the alternative definition of an unstable or marginally stable photon surface in Theorem~\ref{thm:stbsig}.
Remarkably, the inequality~\eqref{A3:ineq} states that $S$ is a marginally stable photon surface only if $\Delta_{S}P|_{p}=0$, which is fulfilled if the matter distribution on $S$ in the vicinity of $p$ is uniform in the angular directions to the second order.
Now we have arrived at the sonic point/photon surface correspondence:

\begin{thm}[Sonic point/photon surface correspondence for steady rotating flow (including Refs.~\cite{Koga:2016jjq,Koga:2018ybs,Koga:2019teu})]
Consider a (d+1)-dimensional spacetime equipped with the metric~\eqref{met} with \eqref{max}, a steady rotating flow of the radiation fluid, which admits the $t$- and $\phi$-translational symmetries and the reversal symmetries in the other angular directions together with the additional condition~\eqref{add}, and the observer on the equatorial plane co-rotating with the flow.
If the flow is of class $C^{2}$ at the sonic point on the equatorial plane, the sonic point of the flow on the equatorial plane must be on an $r=const.$ photon surface that is either unstable or marginally stable.
\label{thm:2}
\end{thm}

\section{summary}
\label{IV}

We considered a transonic steady perfect fluid flow 
associated 
with fiducial observers in a $(d+1)$-dimensional  general stationary spacetime. 
The stationary flow of perfect fluid and fiducial observers are defined 
in terms of the timelike Killing vector field. The speed of the fluid flow is 
defined by the fluid velocity relative to the fiducial observer.
In the neighborhood of each point, 
we 
defined streamlines as the spacelike integral curves 
given by
the projection of the fluid $(d+1)$-velocity onto the spatial direction orthogonal to 
the Killing vector field. 
We showed that the congruence of streamlines must satisfy the equation~\eqref{con1} and the inequality~\eqref{con2} 
at the sonic point, both of which are conditions for the function $\Phi$ defined in Eq.~\eqref{Phi0}. 
The first condition~\eqref{con1} comes from that the speeds of the fluid
and sound must be equal to each other, 
and the second condition~\eqref{con2} ensures that the solution exists at least in the vicinity of the sonic point.
It is worth mentioning that, in the phase space analysis, those conditions require that the point in the phase space to be the saddle point of the Hamiltonian, as we show in Appendix~\ref{app:phase}.
Furthermore, it is remarkable that $\Phi$ reduces to the time-time component of the shear tensor for the streamline congruence at the sonic point if the speed of sound at the sonic point is given by $1/\sqrt{d}$, which leads to Proposition~\ref{prop:5}.
In particular, 
$\Phi$ is the time-time component of the shear tensor at every point if we consider 
the radiation fluid, 
whose speed of sound is given by $1/\sqrt{d}$, 
and the necessary conditions are reduced to 
Eqs.~\eqref{con1:rad} and \eqref{con2:rad} 
as in Theorem~\ref{prop:2}.

Then we considered a photon surface, which is defined as a totally umbilical non-spacelike hypersurface, as the proper section 
at 
the sonic point, 
where the proper section is 
defined as the codimension-1 foliation satisfying the property given in Definition\ref{def:section} 
associated with a streamline of the fluid flow congruence. 
A timelike photon surface, 
as well as the locus of the sonic point 
for the steady radiation fluid flow, is 
defined in
terms of the trace-free part of the second fundamental form.  
Therefore, the shear tensor for the 
proper section 
and the trace-free part of the second fundamental form 
of 
the photon surface 
are essential in the sonic point/photon surface correspondence. 
We found that 
the first condition~\eqref{con1} for the sonic point 
does not imply that all the components of the trace-free part of the second fundamental form vanish, i.e. the sonic point is not necessarily umbilical, which implies that the sonic point/photon surface correspondence requires 
spatial symmetry.
Actually, in the known
results~\cite{Koga:2016jjq,Koga:2018ybs,Koga:2019teu} of the sonic
point/photon surface correspondence, the spacetime admits spatial
spherical, planar or hyperbolic symmetry together with staticity, 
and the symmetry eliminates redundant degrees of freedom in the condition for the photon surface. 
We also showed that the second condition for the sonic point, obtained as an inequality, 
can be rewritten as an inequality for the pressure of the fluid as a function on the photon surface. 

Finally, we presented the sonic point/photon surface correspondence in terms of the shear tensor based on Refs.~\cite{Koga:2016jjq,Koga:2018ybs,Koga:2019teu} 
showing that 
the sonic point is on a photon surface of constant $r$ if and only if the speed of sound at the sonic point is given by $1/\sqrt{d}$ for any perfect fluid.
This implies that the perfect fluid with any equation of state necessarily behaves like the radiation fluid on the sonic points.

Last but not least, we 
assumed 
at least $C^{1}$ smoothness at the sonic point on the flow throughout this paper, 
and excluded the possibility of the transonic shock at the sonic point. 
Therefore, 
a weak solution in the neighborhood of the sonic point 
has not been taken into account
in our analysis. 
If we allow the existence of the shock at the sonic point in Theorem~\ref{thm:2}, 
there might be a sonic point on a stable photon surface or even off the photon surfaces.

\section*{Acknowledgements}
We are grateful to Y. Katou for useful discussion.
This work was supported by the JSPS Grant-in-Aid for Scientific Research No.~JP19H01895 (C.Y. and T.H.) and JP19K03876 (T.H.) and the JSPS Grant-in-Aid for JSPS Fellows No.~JP19J12007 (Y.K.).
\appendix

\section{Deformation of the energy-momentum conservation law}

In this appendix, we perform the deformation of the energy-momentum conservation law~\eqref{energy} for the steady perfect fluid flow. 
For a vector field $\bm X $ on $U$, the energy-momentum conservation law~\eqref{energy} contracted with 
$\bm X$ is given by
\begin{equation}
\nabla\cdot[nh\bm u(\bm u\cdot\bm X)]-nh\bm u\cdot\nabla_{\bm u}\bm X+{\cal L}_{\bm X}P=0.
\label{X0}
\end{equation}
Applying the continuity equation~\eqref{cont} to the first term of Eq.~\eqref{X0} gives
\begin{equation}
n{\cal L}_{\bm u}[h(\bm u\cdot\bm X)]-nh\bm u\cdot\nabla_{\bm u}\bm X+{\cal L}_{\bm X}P=0.
\label{X}
\end{equation}

\subsection{The $\bar{\bm\eta}$ component}
\label{app:de laval}

Substituting $\bm X=\bar{\bm\eta}$ into Eq.~\eqref{X} yields
\begin{equation}
n{\cal L}_{\bm u}\left[h(\bm u\cdot\bar{\bm\eta})\right]-nh\bm u\cdot\nabla_{\bm u}\bar{\bm\eta}+{\cal L}_{\bar{\bm\eta}}P=0.
\label{X:eta}
\end{equation}
We rewrite the first two terms of Eq.~\eqref{X:eta} by applying the orthogonal decomposition~\eqref{decomp} of $\bm u$.
The first term is deformed to be
\begin{align}
{\cal L}_{\bm u}\left[h(\bm u\cdot\bar{\bm\eta})\right]
&=\left(\bm u\cdot\bar{\bm\eta}\right){\cal L}_{\bar{\bm\eta}}\left[h(\bm u\cdot\bar{\bm\eta})\right]\notag\\
&=\left(\bm u\cdot\bar{\bm\eta}\right)^{2}\Big\{{\cal L_{\bar{\bm\eta}}}h+h{\cal L}_{\bar{\bm\eta}}\Big(\ln|\bm u\cdot\bar{\bm\eta}|\Big)\Big\}.
\label{X:eta:1term}
\end{align}
From Eq.~\eqref{v}, the overall factor of Eq.~\eqref{X:eta:1term} is rewritten as $\left(\bm u\cdot\bar{\bm\eta}\right)^{2}=v^{2}\left(1-v^{2}\right)^{-1}$, so that we obtain
\begin{equation}
{\cal L}_{\bm u}[h(\bm u\cdot\bar{\bm\eta})]=v^{2}(1-v^{2})^{-1}\Big\{{\cal L_{\bar{\bm\eta}}}h+h{\cal L}_{\bar{\bm\eta}}\Big(\ln|\bm u\cdot\bar{\bm\eta}|\Big)\Big\}.
\label{X:eta1}
\end{equation}
The second term of Eq.~\eqref{X:eta} is deformed to be
\begin{align}
\bm u\cdot\nabla_{\bm u}\bar{\bm\eta}&=\bm B(\bm u,\bm u)\notag\\
&=\left(1-v^{2}\right)^{-1}\Big\{\bm B(\bar{\bm\xi},\bar{\bm\xi})
+v\bm B(\bar{\bm\xi},\bar{\bm\eta})+v\bm B(\bar{\bm\eta},\bar{\bm\xi})+v^{2}\bm B(\bar{\bm\eta},\bar{\bm\eta})\Big\}\notag\\
&=\left(1-v^{2}\right)^{-1}\bm B(\bar{\bm\xi},\bar{\bm\xi}), 
\label{X:eta22}
\end{align}
where we used Eqs.~\eqref{B:1} and \eqref{B:3}.

We substitute Eqs.~\eqref{X:eta1} and \eqref{X:eta22} into Eq.~\eqref{X:eta}, and we also apply the first law of thermodynamics~\eqref{1st} to the last term of Eq.~\eqref{X:eta}: 
\begin{equation}
n\left(1-v^{2}\right)^{-1}\left\{{\cal L}_{\bar{\bm\eta}}h+hv^{2}{\cal L}_{\bar{\bm\eta}}\Big(\ln|\bm u\cdot\bar{\bm\eta}|\Big)-h\bm B(\bar{\bm\xi},\bar{\bm\xi})\right\}-nT{\cal L}_{\bar{\bm\eta}}s=0.
\label{X:eta3}
\end{equation}
Here, from the isentropic condition~\eqref{ise} and the stationarity~\eqref{steady1} of the flow, we obtain 
\begin{equation}
{\cal L}_{\bar{\bm\eta}}s=0,
\end{equation}
i.e. the last term of Eq.~\eqref{X:eta3} vanishes.
We also deform the first term of Eq.~\eqref{X:eta3} using the continuity equation~\eqref{cont} as follows.
The continuity equation~\eqref{cont} is rewritten as ${\cal L}_{\bm u}n=-n\nabla\cdot\bm u$, 
so that we have
\begin{align}
{\cal L}_{\bar{\bm\eta}}h&=hn^{-1}v_{\rm s}^{2}{\cal L}_{\bar{\bm\eta}}n\notag\\
&=-hv_{\rm s}^{2}\nabla\cdot\bm u\notag\\
&=-hv_{\rm s}^{2}\left\{{\cal L}_{\bar{\bm\eta}}\Big(\ln|\bm u\cdot\bar{\bm\eta}|\Big)+\Theta\right\}.
\label{X:eta4}
\end{align}
Substituting Eq.~\eqref{X:eta4} into Eq.~\eqref{X:eta3} gives
\begin{equation}
-nh(1-v^{2})^{-1}\left\{(v_{\rm s}^{2}-v^{2}){\cal L}_{\bar{\bm\eta}}\Big(\ln|\bm u\cdot\bar{\bm\eta}|\Big)+\bm B(\bar{\bm\xi},\bar{\bm\xi})+v_{\rm s}^{2}~\Theta\right\}=0.
\end{equation}
Because we have supposed $nh(1-v^{2})^{-1}\neq0$, 
we arrive at the de Laval nozzle-like equation~\eqref{de laval}.

\subsection{The other spatial components}
\label{app:ei}

Let $\bm e_{i}$ be another orthonormal vector field on $U$ than $\bar{\bm\xi}$ and $\bar{\bm\eta}$, satisfying
\begin{equation}
\bar{\bm\xi}\cdot\bm e_{i}=\bar{\bm\eta}\cdot\bm e_{i}=0.
\label{ei:ortho}
\end{equation}
Substituting $\bm X=\bm e_{i}$ into Eq.~\eqref{X} yields 
\begin{equation}
n{\cal L}_{\bm u}[h(\bm u\cdot\bm e_{i})]-nh\bm u\cdot\nabla_{\bm u}\bm e_{i}+{\cal L}_{\bm e_{i}}P=0.
\label{X:ei}
\end{equation}
From the expression~\eqref{decomp} of $\bm u$ and the orthogonal relation~\eqref{ei:ortho}, we find $\bm e_{i}$ is also orthogonal to $\bm u$, i.e. $\bm u\cdot\bm e_{i}=0$.
Therefore, the first term of Eq.~\eqref{X:ei} vanishes.
We deform the second term of Eq.~\eqref{X:ei} as 
\begin{align}
-\bm u\cdot\nabla_{\bm u}\bm e_{i}&=\nabla_{\bm u}\bm u\cdot\bm e_{i}\notag\\
&=\left(1-v^{2}\right)^{-1}\Big\{\nabla_{\bar{\bm\xi}}\bar{\bm\xi}+v\Big(\nabla_{\bar{\bm\xi}}\bar{\bm\eta}+\nabla_{\bar{\bm\eta}}\bar{\bm\xi}\Big)+v^{2}\nabla_{\bar{\bm\eta}}\bar{\bm\eta}\Big\}\cdot\bm e_{i}.
\label{X:ei1}
\end{align}
Applying the Killing equation $\nabla_{\mu}\xi_{\nu}+\nabla_{\nu}\xi_{\mu}=0$ and the orthogonal relation~\eqref{ei:ortho}, we rewrite the first term and the third term of Eq.~\eqref{X:ei1} as follows:
\begin{align}
\nabla_{\bar{\bm\xi}}\bar{\bm\xi}\cdot\bm e_{i}&=\left(\bm\xi\cdot\bm\xi\right)^{-1}\nabla_{\bm e_{i}}\bm\xi\cdot\bm\xi\notag\\
&=\frac{1}{2}\left(\bm\xi\cdot\bm\xi\right)^{-1}\nabla_{\bm e_{i}}\left(\bm\xi\cdot\bm\xi\right)\notag\\
&={\cal L}_{\bm e_{i}}\left(\ln\sqrt{|\bm\xi\cdot\bm\xi|}\right),
\label{X:ei:second}
\end{align}
and
\begin{align}
\nabla_{\bar{\bm\eta}}\bar{\bm\xi}\cdot\bm e_{i}&=-\left(\bm\xi\cdot\bm\xi\right)^{-1/2}\nabla_{\bm e_{i}}\bm\xi\cdot\bar{\bm\eta}\notag\\
&=\nabla_{\bm e_{i}}\bar{\bm\eta}\cdot\bar{\bm\xi}\notag\\
&=\bm B(\bar{\bm\xi},\bm e_{i}).
\label{X:ei:third}
\end{align}
Eq.~\eqref{X:ei1} transforms to
\begin{equation}
-\bm u\cdot\nabla_{\bm u}\bm e_{i}=\left(1-v^{2}\right)^{-1}\Big\{{\cal L}_{\bm e_{i}}\left(\ln\sqrt{|\bm\xi\cdot\bm\xi|}\right)+2v\bm B(\bar{\bm\xi},\bm e_{i})+v^{2}\bm a\cdot\bm e_{i}\Big\}.
\end{equation}
Therefore, Eq.~\eqref{X:ei} results in
\begin{equation}
nh\left(1-v^{2}\right)^{-1}\Big\{{\cal L}_{\bm e_{i}}\left(\ln\sqrt{|\bm\xi\cdot\bm\xi|}\right)+2v\bm B(\bar{\bm\xi},\bm e_{i})+v^{2}\bm a\cdot\bm e_{i}\Big\}+{\cal L}_{\bm e_{i}}P=0.
\label{energy:ei0}
\end{equation}
Here, we denote the second term of Eq.~\eqref{energy:ei0} by the following vector field $\bm b$ on $U$ orthogonal to both of $\bar{\bm\xi}$ and $\bar{\bm\eta}$:
\begin{equation}
\bm b:=\sum_{i}\bm B(\bar{\bm\xi},\bm e_{i})\bm e_{i}.
\end{equation}
Because 
we have supposed $nh\left(1-v^{2}\right)^{-1}\neq0$, Eq.~\eqref{energy:ei0} is rewritten as 
\begin{equation}
\bm e_{i}\cdot\left\{\dd\left(\ln\sqrt{|\bm\xi\cdot\bm\xi|}\right)+2v\,\bm b+v^{2}\bm a+\left(1-v^{2}\right)(nh)^{-1}\dd P\right\}=0.
\label{energy:ei}
\end{equation}

\section{Phase space analysis}
\label{app:phase}

We derived the necessary conditions~\eqref{con1} and \eqref{con2} for the steady perfect fluid flow 
at the sonic point.
These conditions have been conventionally derived through the phase space analysis.
In this appendix, we revisit Theorem~\ref{prop:con} by considering the Hamiltonian mechanics on the streamlines.

\subsection{Conserved quantities along streamlines}

Let us consider the steady perfect fluid flow 
in $U$.
The Lie-commutativity~\eqref{frob} stipulates that $\bm\xi$ and $\bar{\bm\eta}$ span a foliation of $U$.
Since we can take $U$ to be an 
arbitrarily 
small region including a point $p$, 
we can 
always find a 
local chart $(U,(\tau,\lambda,X^{i}))$ by choosing a suitable $U$, where the coordinates $\tau$ and $\lambda$ on each leaf are respectively the parameters on the integral curves of $\bm\xi$ and $\bar{\bm\eta}$ satisfying
\begin{subequations}
\begin{eqnarray}
\bm\xi&=&\bm\partial_{\tau},\label{xi tau}\\
\bar{\bm\eta}&=&\bm\partial_{\lambda},\label{eta lambda}
\end{eqnarray}
\end{subequations}
and the other coordinates $X^{i}$ are $d-1$ independent 
functions taking constant on each leaf.
$\bm\xi$ and $\bar{\bm\eta}$ are holonomic bases on each leaf in this coordinate system.

Because the fluid $(d+1)$-velocity $\bm u$ has been expressed as the combination of $\bm\xi$ and $\bar{\bm\eta}$ from Eq.~\eqref{decomp}, $\bm u$ is tangent to the leaves of the foliation of $U$.
In other words, for any integral curve of $\bm u$ in $U$, there exists a leaf $\Sigma\subset U$ of the foliation 
in a manner such that the integral curve is given as the image $|\gamma|$ of 
$\gamma:(\sigma_{1},\sigma_{2})\to\Sigma$, where $\sigma\in(\sigma_{1},\sigma_{2})$ is the proper time along $|\gamma|$, and $\bm u$ is expressed as $\bm u=\dd/\dd\sigma$.
We get the parametric representation of each integral curve $|\gamma|$ of $\bm u$ as 
\begin{subequations}
\begin{eqnarray}
\tau&=&\tau(\sigma),\\
\lambda&=&\lambda(\sigma),\\
X^{i}&=&const.,
\end{eqnarray}
\end{subequations}
and the coordinate representation of $\bm u=\dd/\dd\sigma$ is given by
\begin{equation}
\bm u=u^{\tau}\bm\partial_{\tau}+u^{\lambda}\bm\partial_{\lambda},
\label{u:coord}
\end{equation}
where $u^{\tau}:={\rm d}\tau/{\rm d}\sigma$ and $u^{\lambda}:={\rm d}\lambda/{\rm d}\sigma$.
The coefficients $u^{\tau}$ and $u^{\lambda}$ are related to the other expressions as follows:
\begin{subequations}
\begin{eqnarray}
-\bm u\cdot\bar{\bm\xi}&=&\frac{1}{\sqrt{1-v^{2}}}\,=\,u^{\tau}\sqrt{|g_{\tau\tau}|},\\
\bm u\cdot\bar{\bm\eta}&=&\frac{v}{\sqrt{1-v^{2}}}\,=\,u^{\lambda}.
\label{ulam}
\end{eqnarray}
\end{subequations}
The normalization condition~\eqref{norm} for $\bm u$ is deformed to
\begin{equation}
g_{\tau\tau}(u^{\tau})^{2}+(u^{\lambda})^{2}=-1,
\label{norm0}
\end{equation}
where we used that $\bar{\bm\eta}=\bm\partial_{\lambda}$ is a unit vector field, i.e. $g_{\lambda\lambda}=1$.

In the coordinate system $(\tau,\lambda,X^{i})$, the continuity equation~(\ref{cont}) and 
the energy-momentum conservation law~(\ref{energy}) contracted with $\bm\xi$ are rewritten as
\begin{subequations}
\begin{eqnarray}
\partial_{\mu}\left(\sqrt{|g|}\,nu^{\mu}\right)&=&0,\label{acc0}\\
\partial_{\mu}\left(\sqrt{|g|}\,nhg_{\tau\tau}u^{\tau}u^{\mu}\right)&=&0,\label{nu0}
\end{eqnarray}
\end{subequations}
where $g$ denotes the determinant of $\bm g$ in the  coordinate system $(\tau,\lambda,X^{i})$.
We find the associated conserved quantities $\mu$ and $\nu$ taking constant on each leaf as follows: 
\begin{subequations}
\begin{eqnarray}
\mu=\mu(X^{i})&:=&\sqrt{|g|}\,nu^{\lambda},
\label{acc}\\
\nu=\nu(X^{i})&:=&\sqrt{|g|}\,nhg_{\tau\tau}u^{\tau}u^{\lambda}.
\label{nu}
\end{eqnarray}
\end{subequations}
In addition, the specific entropy $s$ is the third conserved quantity on each leaf from the isentropic condition~\eqref{ise}.
In this paper, we use $s$, $\mu$ and the ratio 
\begin{equation}
\varepsilon=\varepsilon(X^{i}):=\frac{\nu}{\mu}=g_{\tau\tau}hu^{\tau}, 
\label{varep}
\end{equation}
that is also a conserved quantity instead of $\nu$.

The conserved quantity $\mu$ on each leaf of the foliation is what is called the accretion rate in accretion problems.
For $\mu=0$, we get $u^{\lambda}=0$ from Eq.~\eqref{acc}, and we further get $v=0$ from Eq.~\eqref{ulam}.
Therefore, the flow is at rest in the reference frame of the observer if $\mu=0$, which is consistent with the picture of non-accreting fluid.
Note that $\varepsilon$ in Eq.~\eqref{varep} takes non-zero value even if $\mu=0$.

In the coordinate system $(\tau,\lambda,X^{i})$, the de Laval nozzle equation~\eqref{de laval} is rewritten as
\begin{equation}
\left(v_{\rm s}^{2}-v^{2}\right)\partial_{\lambda}\left(\ln |u^{\lambda}|\right)+\Phi=0,
\label{laval}
\end{equation}
where we can express the second term $\Phi$ as
\begin{equation}
\Phi=-\partial_{\lambda}\left(\ln\sqrt{|g_{\tau\tau}|}\right)+v_{\rm s}^{2}~\partial_{\lambda}\left(\ln\sqrt{|g|}\right)
\label{Phi}
\end{equation}
using
\begin{subequations}
\begin{eqnarray}
\Theta&=&\partial_{\lambda}\left(\ln\sqrt{|g|}\right),\label{exp}\\
B_{\tau\tau}&=&\partial_{\lambda}\left(\ln\sqrt{|g_{\tau\tau}|}\right)g_{\tau\tau}.\label{Btt}
\end{eqnarray}
\end{subequations}

\subsection{Hamiltonian mechanics on streamlines}
\label{II3}

Once the foliation of $U$ is given, and the values of the conserved
quantities $s$, $\mu$ and $\varepsilon$ on one of the leaves $\Sigma$
are fixed, then Eqs.~\eqref{norm}, \eqref{eos}, \eqref{acc} and
\eqref{nu} result in the simultaneous algebraic equations for
$u^{\tau}(\sigma)$, $u^{\lambda}(\sigma)$, $n(\sigma)$ and $h(\sigma)$,
where we may rewrite the equation of state~\eqref{eos} to $h=h(n,s)$
since we have supposed that the squared speed of
sound~\eqref{vs} is positive, i.e. $v_{\rm s}^{2}=h^{-1}(\partial
P/\partial n)_{s}>0$, and therefore $n=n(P,s)$ is a strictly
monotonically increasing function of $P$ if $s$ is fixed.
By eliminating $u^{\tau}$, $u^{\lambda}$ and $h$ from Eqs.~\eqref{norm}, \eqref{eos}, \eqref{acc} and \eqref{nu}, we obtain the following algebraic equation for $n=n(\sigma)$:
\begin{equation}
\varepsilon^{2}=h(n(\sigma),s)^{2}|g_{\tau\tau}(\sigma)|\left[1+\left(\frac{\mu}{\sqrt{|g(\sigma)|}~n(\sigma)}\right)^{2}\right].
\label{Hamiltonian}
\end{equation}
We define the product space $\Gamma:=(\sigma_{1},\sigma_{2})\times\mathbb{R}_{>0}$ of the proper time $\sigma$ and the number density $n$, and express the right-hand side of Eq.~\eqref{Hamiltonian} as
\begin{equation}
F(\sigma,n):=h(n,s)^{2}|g_{\tau\tau}(\sigma)|\left[1+\left(\frac{\mu}{\sqrt{|g(\sigma)|}~n}\right)^{2}\right].
\label{F}
\end{equation}
Eq.~\eqref{Hamiltonian} is simply rewritten as $\varepsilon^{2}=F(\sigma,n)$, and we obtain $n=n(\sigma)$ as a level curve of the 2-dimensional surface on $\Gamma$ that is given by $F:\Gamma\to\mathbb{R}_{>0}$ from Eq.~\eqref{F}.
The level curve 
can be also given as the image $|c|$ of the map $c:(\sigma_{1},\sigma_{2})\to\Gamma$.

It is well known that the level curve of the 2-dimensional surface can
be regarded as the orbit of the Hamiltonian mechanics in 2-dimensional
phase space in general, and the investigation of fluid systems in terms
of the Hamiltonian mechanics, referred to as the phase space analysis,
is widely used (see, e.g. Ref.~\cite{Chaverra:2015bya}). 
We can actually derive the canonical equations for $(\sigma(\tilde\lambda),n(\tilde{\lambda}))$ as follows, where $\tilde{\lambda}$ is the parameter on the level curve $|c|$.
Because $|c|$ is a level curve of $F$, the directional derivative of $F$ along $|c|$ vanishes.  
We denote the directional derivative along $|c|$ by ${\rm d}_{|c|}$, and we get
\begin{equation}
{\rm d}_{|c|} F=(\partial_{\sigma}F){\rm d}_{|c|}\sigma+(\partial_{n}F){\rm d}_{|c|} n=0.
\label{c}
\end{equation}
In the parametric representation $(\sigma(\tilde\lambda),n(\tilde\lambda))$, Eq.~\eqref{c} is deformed to the canonical equations:
\begin{equation}
\frac{{\rm d}}{\rm d \tilde\lambda}
\begin{pmatrix}\sigma\\n\end{pmatrix}=N\begin{pmatrix}\partial_{n}F\\-\partial_{\sigma}F\end{pmatrix},
\label{canonical}
\end{equation}
where $N$ is the Lagrange multiplier associated with the parametrization $\tilde\lambda$ of $|c|$.
In the context of the phase space analysis, $\Gamma$ is the phase space, $F$ is the Hamiltonian, $|c|$ is the orbit, and $n$ is  the conjugate momentum to the position $\sigma$.

We also redefine the speed of sound and the speed of the flow as functions $v_{\rm s}=v_{\rm s}(n)$ and $v=v(\sigma,n)$ on $\Gamma$ through Eqs.~\eqref{vs} and \eqref{v}.
From Eq.~\eqref{v}, $v=v(\sigma,n)$ is given by
\begin{equation}
\left(1-v^{2}\right)^{-1}=1+\left(\frac{\mu}{\sqrt{|g|}\,n}\right)^{2}.
\label{v:phase}
\end{equation}
From Eq.~\eqref{v:phase},
\begin{subequations}
\begin{eqnarray}
\partial_{n}v^{2}&=&-2v^{2}(1-v^{2})n^{-1},\label{v2_n}\\
\partial_{\sigma}v^{2}&=&-2v^{2}(1-v^{2})\partial_{\sigma}\left(\ln\sqrt{|g|}\right).\label{v2_sigma}
\end{eqnarray}
\end{subequations}
Substituting Eq.~\eqref{v:phase} into Eq.~\eqref{F}, we rewrite $F$ to
\begin{equation}
F=h^{2}|g_{\tau\tau}|\left(1-v^{2}\right)^{-1}.
\end{equation}
Applying Eqs.~\eqref{v2_n} and \eqref{v2_sigma}, we obtain
\begin{subequations}
\begin{eqnarray}
\partial_{n}F&=&2Fn^{-1}\left(v_{\rm s}^{2}-v^{2}\right),\label{F:n}\\
\partial_{\sigma}F&=&2F\left\{-\Phi_{\Gamma}+\left(v_{\rm s}^{2}-v^{2}\right)\partial_{\sigma}\left(\ln\sqrt{|g|}\right)\right\},
\label{F:s}
\end{eqnarray}
\end{subequations}
where
\begin{equation}
\Phi_{\Gamma}(\sigma,n):=-\partial_{\sigma}\left(\ln\sqrt{|g_{\tau\tau}|}\right)+v_{\rm s}^{2}~\partial_{\sigma}\left(\ln\sqrt{|g|}\right).
\label{Phic}
\end{equation}
We find $\Phi$ and $\Phi_{\Gamma}$ obey $\Phi=u^{\lambda}\Phi_{\Gamma}$ from Eqs.~\eqref{Phi} and \eqref{Phic}.

\subsection{Sonic point in the spacetime and saddle point in the phase space}

Now we present the following theorem which is consistent with Theorem~\ref{prop:con}. 
\begin{thm}
Let $\gamma:(\sigma_{1},\sigma_{2})\to U$ be the fluid world line, the integral curve of $\bm u$, including $p$.
Consider the Hamiltonian mechanics on the streamline $|\gamma|$, where the phase space is given as $\Gamma=(\sigma_{1},\sigma_{2})\times\mathbb{R}_{>0}$, the space of $\sigma$ and $n$, and the Hamiltonian $F:\Gamma\to\mathbb{R}_{>0}$ is given by Eq.~\eqref{F}.
Let $c:(\sigma_{1},\sigma_{2})\to\Gamma$ give $n=n(\sigma)$.
If $p\in|\gamma|$ is a sonic point, $p_{c}\in|c|$ defined by $p_{c}=c\circ\gamma^{-1}(p)$ must be a saddle point of $F$, i.e.
\begin{subequations}
\begin{eqnarray}
\partial_{n}F|_{p_{c}}&=&0,\label{cr1}\\
\partial_{\sigma}F|_{p_{c}}&=&0,\label{cr2}
\end{eqnarray} 
\end{subequations}
and the Hessian at $p_{c}$ defined by
\begin{equation}
{\rm Hess}:=\Big\{(\partial_{\sigma}^{2}F)(\partial_{n}^{2}F)-(\partial_{\sigma}\partial_{n}F)^{2}\Big\}\Big|_{p_{c}}
\label{hess0}
\end{equation}
obeys ${\rm Hess}\le0$.
\label{prop:8}
\end{thm}
\begin{proof}
The point $p\in|\gamma|$ in the spacetime is mapped to the point $p_{c}\in|c|$ in the phase space by the composition $c\circ \gamma^{-1}$.
If $p$ is a sonic point, simultaneously the following equation in $\Gamma$ must be satisfied:
\begin{equation}
\left(v_{\rm s}^{2}-v^{2}\right)\Big|_{p_{c}}=0.
\label{vs:phase}
\end{equation}
From Eq.~\eqref{F:n}, Eq.~\eqref{vs:phase} is equivalent to Eq.~\eqref{cr1}.
In addition, from the canonical equations~\eqref{canonical}, Eq.~\eqref{cr1} is also equivalent to
\begin{equation}
\frac{{\rm d}\sigma}{{\rm d}\tilde\lambda}\Big|_{p_{c}}=0.
\end{equation}
where we set $N$ to $N=1$.
In the context of dynamical systems, the set of points where ${\rm d}\sigma/{\rm d}\tilde\lambda=0$ with $N=1$ is called the $\sigma$-nullcline.
Therefore, we can conclude that $p_{c}$ must be a point on the $\sigma$-nullcline $\{\partial_{n}F=0\}$ if $p$ is the sonic point.

We find from the canonical equation~\eqref{canonical} with $N=1/(\partial_{n}F)$ which leads to $\tilde\lambda=\sigma$ that the gradient of $n=n(\sigma)$ is given by
\begin{equation}
\frac{{\rm d} n}{\rm d \sigma}=-\frac{\partial_{\sigma}F}{\partial_{n}F}.
\label{dnds}
\end{equation}
We have supposed that the flow is of class $C^{2}$ in $U$, and therefore Eq.~\eqref{cr2} must also be satisfied in order that Eq.~\eqref{dnds} does not diverge at $p_{c}$.
So far we showed that both Eqs.~\eqref{cr1} and \eqref{cr2} must be satisfied at $p_{c}$, i.e. $p_{c}$ must be a critical point of $F$.
Substituting Eq.~\eqref{F:s} into Eq.~\eqref{cr2}, we obtain 
\begin{equation}
\Phi_{\Gamma}|_{p_{c}}=0.
\label{PhiGamma0}
\end{equation}
Eq.~\eqref{PhiGamma0} is equivalent to Eq.~\eqref{con1} in Theorem~\ref{prop:con}.

We then evaluate the Hessian at $p_{c}$ defined in Eq.~\eqref{con2}.
The second derivatives of $F=F(\sigma,n)$ at the critical point $p_{c}$ is given as follows:
\begin{subequations}
\begin{eqnarray}
\partial_{n}^{2}F|_{p_{c}}&=&2Fn^{-1}\partial_{n}\left(v_{\rm s}^{2}-v^{2}\right)\Big|_{p_{c}},
\label{nn2}\\
\partial_{\sigma}\partial_{n}F|_{p_{c}}&=&-2F(\partial_{n}v^{2})\partial_{\sigma}\left(\ln\sqrt{|g|}\right)\Big|_{p_{c}},\\
\partial_{\sigma}^{2}F|_{p_{c}}&=&2F\left\{-\partial_{\sigma}\Phi_{\Gamma}-n\left(\partial_{n}v^{2}\right)\left[\partial_{\sigma}\left(\ln\sqrt{|g|}\right)\right]^{2}\right\}\Big|_{p_{c}}.
\end{eqnarray}
\end{subequations}
Eq.~\eqref{hess0} is calculated to be
\begin{align}
{\rm Hess}&=2F\left\{-(\partial_{n}^{2}F)(\partial_{\sigma}\Phi_{\Gamma})-2F\left(\partial_{n}v_{\rm s}^{2}\right)\left(\partial_{n}v^{2}\right)\left[\partial_{\sigma}\left(\ln\sqrt{|g|}\right)\right]^{2}\right\}\Big|_{p_{c}}\notag\\
&=2F\left\{-(\partial_{n}^{2}F)(\partial_{\sigma}\Phi_{\Gamma})-2Fn^{-1}\left(\partial_{n}v_{\rm s}^{2}\right)\left(\partial_{\sigma}v^{2}\right)~\partial_{\sigma}\left(\ln\sqrt{|g|}\right)\right\}\Big|_{p_{c}}.
\label{hess}
\end{align}
We rewrite all terms of Eq.~\eqref{hess} with the directional derivative ${\rm d}/{\rm d}\sigma$ along the level curve $|c|$ by using
\begin{subequations}
\begin{eqnarray}
\partial_{\sigma}\Phi_{\Gamma}&=&\frac{{\rm d}\Phi_{\Gamma}}{{\rm d}\sigma}-\frac{{\rm d}v_{\rm s}^{2}}{{\rm d}\sigma}~\partial_{\sigma}\left(\ln\sqrt{|g|}\right),
\label{Ac}\\
\left(\partial_{n}v_{\rm s}^{2}\right)\left(\partial_{\sigma}v^{2}\right)
&=&-\left(\partial_{n}v_{\rm s}^{2}\right)\frac{{\rm d}}{{\rm d}\sigma}\left(v_{\rm s}^{2}-v^{2}\right)+\frac{{\rm d}v_{\rm s}^{2}}{{\rm d}\sigma}~\partial_{n}\left(v_{\rm s}^{2}-v^{2}\right).
\label{vv}
\end{eqnarray}
\end{subequations}
The Hessian~\eqref{hess} results in
\begin{align}
{\rm Hess}&=2F\left\{-(\partial_{n}^{2}F)\frac{{\rm d}\Phi_{\Gamma}}{{\rm d}\sigma}+2Fn^{-1}\left(\partial_{n}v_{\rm s}^{2}\right)\frac{{\rm d}}{{\rm d}\sigma}\left(v_{\rm s}^{2}-v^{2}\right)\partial_{\sigma}\left(\ln\sqrt{|g|}\right)\right\}\Big|_{p_{c}}\notag\\
&=2F\left\{-(\partial_{n}^{2}F)\frac{{\rm d}\Phi_{\Gamma}}{{\rm d}\sigma}-Fn^{-1}v^{-2}\left(1-v^{2}\right)^{-1}\left(\partial_{n}v_{\rm s}^{2}\right)\left(\partial_{\sigma}v^{2}\right)\frac{{\rm d}}{{\rm d}\sigma}\left(v_{\rm s}^{2}-v^{2}\right)\right\}\Big|_{p_{c}}.
\label{hess2}
\end{align}
Eq.~\eqref{vv} is also rewritten as
\begin{equation}
\left(\partial_{n}v_{\rm s}^{2}\right)\left(\partial_{\sigma}v^{2}\right)
=-\left(\partial_{n}v^{2}\right)\frac{{\rm d}}{{\rm d}\sigma}\left(v_{\rm s}^{2}-v^{2}\right)+\frac{{\rm d}v^{2}}{{\rm d}\sigma}~\partial_{n}\left(v_{\rm s}^{2}-v^{2}\right).
\label{vv2}
\end{equation}
Substituting Eq.~\eqref{vv2} into the Hessian~\eqref{hess2} leads to
\begin{eqnarray}
{\rm Hess}&=&2F\Biggl\{-(\partial_{n}^{2}F)\left[\frac{{\rm d}\Phi_{\Gamma}}{{\rm d}\sigma}+\frac{1}{2}v^{-2}\left(1-v^{2}\right)^{-1}\frac{{\rm d}}{{\rm d}\sigma}\left(v_{\rm s}^{2}-v^{2}\right)\frac{{\rm d}v^{2}}{{\rm d}\sigma}\right]\cr
&&+Fn^{-1}v^{-2}\left(1-v^{2}\right)^{-1}\left(\partial_{n}v^{2}\right)\left[\frac{{\rm d}}{{\rm d}\sigma}\left(v_{\rm s}^{2}-v^{2}\right)\right]^{2}\Biggr\}\Bigg|_{p_{c}}.
\label{hess3}
\end{eqnarray}
Here, the de Laval nozzle-like equation~\eqref{laval} is expressed in $\Gamma$ as
\begin{equation}
\frac{1}{2}v^{-2}\left(1-v^{2}\right)^{-1}\left(v_{\rm s}^{2}-v^{2}\right)\frac{{\rm d}v^{2}}{{\rm d}\sigma}+\Phi_{\Gamma}=0,
\end{equation}
whence we obtain
\begin{equation}
\left\{\frac{1}{2}v^{-2}\left(1-v^{2}\right)^{-1}\frac{{\rm d}}{{\rm d}\sigma}\left(v_{\rm s}^{2}-v^{2}\right)\frac{{\rm d}v^{2}}{{\rm d}\sigma}+\frac{{\rm d}\Phi_{\Gamma}}{{\rm d}\sigma}\right\}\Big|_{p_{c}}=0.
\label{laval:diff}
\end{equation}
Therefore, the first term in the braces of Eq.~\eqref{hess3} vanishes from Eq.~\eqref{laval:diff}, and we obtain 
\begin{equation}
{\rm Hess}=-4F^{2}n^{-2}\left[\frac{{\rm d}}{{\rm d}\sigma}\left(v_{\rm s}^{2}-v^{2}\right)\right]^{2}\Bigg|_{p_{c}}\le0.
\end{equation}
\end{proof}

\section{Inequality at a sonic point}

\subsection{Inequality at a sonic point for radiation fluid}
\label{app:ineq}

In this subsection, we deform the inequality~\eqref{con2:rad} to the inequality~\eqref{ineq} with Eqs.~\eqref{A1}, \eqref{A2} and \eqref{A3}.
The inequality~\eqref{con2:rad} consists of the following three parts:
\begin{equation}
{\cal L}_{\bar{\bm\eta}}\left[\bm\sigma(\bar{\bm\xi},\bar{\bm\xi})\right]\Big|_{p}=C_{1}+C_{2}+C_{3},
\end{equation}
where
\begin{subequations}
\begin{eqnarray}
C_{1}&=&\frac{1}{d}\,{\cal L}_{\bar{\bm\eta}}\Theta\,\Big|_{p},\\
C_{2}&=&(\nabla_{\bar{\bm\eta}}\bm B)(\bar{\bm\xi},\bar{\bm\xi})\,\Big|_{p},\\
C_{3}&=&\Big\{\bm B(\nabla_{\bar{\bm\eta}}\bar{\bm\xi},\bar{\bm\xi})+\bm B(\bar{\bm\xi},\nabla_{\bar{\bm\eta}}\bar{\bm\xi})\Big\}\Big|_{p}.
\end{eqnarray}
\end{subequations}
From Eqs.~\eqref{B:perp1} and \eqref{B:perp2}, the tensor field $\bm B$ has been given as
\begin{equation}
\bm B=\frac{\Theta}{d}\,\bm g^{\perp}+\bm\sigma+\bm\omega+\bar{\bm\eta}^{\flat}\otimes\bm a^{\flat},
\end{equation}
and $\bm B$ obeys
\begin{align}
\nabla_{\bar{\bm\eta}}\bm B&=\bar{\eta}^{\mu}\nabla_{\mu}\nabla_{\nu}\bar{\eta}_{\rho}~\dd x^{\nu}\otimes\dd x^{\rho}\notag\\
&=\Big\{\bar{\eta}^{\mu}[\nabla_{\mu},\nabla_{\nu}]\bar{\eta}_{\rho}+\nabla_{\nu}\left(\bar{\eta}^{\mu}\nabla_{\mu}\bar{\eta}_{\rho}\right)-\left(\nabla_{\nu}\bar{\eta}^{\mu}\right)\left(\nabla_{\mu}\bar{\eta}_{\rho}\right)\Big\}\dd x^{\nu}\otimes\dd x^{\rho}\notag\\
&=-R(\,\cdot\,,\bar{\bm\eta},\,\cdot\,,\bar{\bm\eta})+\nabla\otimes\bm a^{\flat}-\bm B\cdot\bm B.
\label{B:cov}
\end{align}
Contracting indices of Eq.~\eqref{B:cov} with $\bm g$ gives
\begin{align}
{\cal L}_{\bar{\bm\eta}}\Theta&={\rm tr}_{\bm g}(\nabla_{\bar{\bm\eta}}\bm B)\notag\\
&=-{\rm Ric}(\bar{\bm\eta},\bar{\bm\eta})+\nabla\cdot\bm a-\bm B:\bm B,
\label{LTheta}
\end{align}
where the third term of Eq.~\eqref{LTheta} is given as
\begin{equation}
\bm B:\bm B=\frac{\Theta^{2}}{d}+\bm\sigma:\bm\sigma+\bm\omega:\bm\omega.
\end{equation}
We obtain
\begin{equation}
C_{1}=\Big\{-\left(\frac{\Theta}{d}\right)^{2}+\frac{1}{d}\Big(-{\rm Ric}(\bar{\bm\eta},\bar{\bm\eta})+\nabla\cdot\bm a-\bm\sigma:\bm\sigma-\bm\omega:\bm\omega\Big)\Big\}\Big|_{p}.
\label{C1}
\end{equation}
Contracting indices of Eq.~\eqref{B:cov} with $\bar{\bm\xi}$ gives
\begin{equation}
(\nabla_{\bar{\bm\eta}}\bm B)(\bar{\bm\xi},\bar{\bm\xi})=-R(\bar{\bm\xi},\bar{\bm\eta},\bar{\bm\xi},\bar{\bm\eta})+(\nabla\otimes\bm a^{\flat})(\bar{\bm\xi},\bar{\bm\xi})-(\bm B\cdot\bm B)(\bar{\bm\xi},\bar{\bm\xi}),
\label{LBtt}
\end{equation}
where the third term of Eq.~\eqref{LBtt} is calculated as
\begin{equation}
(\bm B\cdot\bm B)(\bar{\bm\xi},\bar{\bm\xi})=-\left(\frac{\Theta}{d}\right)^{2}+\frac{2\Theta}{d}\,\bm\sigma(\bar{\bm\xi},\bar{\bm\xi})+(\bm\sigma\cdot\bm\sigma+\bm\omega\cdot\bm\omega)(\bar{\bm\xi},\bar{\bm\xi}).
\end{equation}
Using the first condition~\eqref{con1:rad}, we obtain
\begin{equation}
C_{2}=\left\{\left(\frac{\Theta}{d}\right)^{2}-R(\bar{\bm\xi},\bar{\bm\eta},\bar{\bm\xi},\bar{\bm\eta})+(\nabla\otimes\bm a^{\flat})(\bar{\bm\xi},\bar{\bm\xi})
-(\bm\sigma\cdot\bm\sigma+\bm\omega\cdot\bm\omega)(\bar{\bm\xi},\bar{\bm\xi})\right\}\,\Bigg|_{p}.
\label{C2}
\end{equation}
$C_{3}$ is deformed to
\begin{align}
C_{3}&=\Big\{2\sqrt{|\bm\xi\cdot\bm\xi|}~\nabla_{\bar{\bm\eta}}\left(\frac{1}{\sqrt{|\bm\xi\cdot\bm\xi|}}\right)\bm B(\bar{\bm\xi},\bar{\bm\xi})
+\frac{1}{\sqrt{|\bm\xi\cdot\bm\xi|}}\left(\bm B(\nabla_{\bar{\bm\eta}}\bm\xi,\bar{\bm\xi})+\bm B(\bar{\bm\xi},\nabla_{\bar{\bm\eta}}\bm\xi)\right)\Big\}\,\Big|_{p}\notag\\
&=\Big\{-2\nabla_{\bar{\bm\eta}}\left(\ln\sqrt{|\bm\xi\cdot\bm\xi|}\right)\bm B(\bar{\bm\xi},\bar{\bm\xi})+\bm B(\nabla_{\bar{\bm\xi}}\bar{\bm\eta},\bar{\bm\xi})+\bm B(\bar{\bm\xi},\nabla_{\bar{\bm\xi}}\bar{\bm\eta})\Big\}\,\Big|_{p}\notag\\
&=\Big\{2\Big[\bm B(\bar{\bm\xi},\bar{\bm\xi})\Big]^{2}+2\left(\bm B\cdot\bm B^{({\rm S})}\right)(\bar{\bm\xi},\bar{\bm\xi})\Big\}\,\Big|_{p}, 
\end{align}
where we used Eqs.~\eqref{B:2} and \eqref{con1}.
Each term is further calculated as
\begin{subequations}
\begin{eqnarray}
\Big[\bm B(\bar{\bm\xi},\bar{\bm\xi})\Big]^{2}&=&\left(-\frac{\Theta}{d}+\bm\sigma(\bar{\bm\xi},\bar{\bm\xi})\right)^{2},\\
(\bm B\cdot\bm B^{({\rm S})})(\bar{\bm\xi},\bar{\bm\xi})&=&-\left(\frac{\Theta}{d}\right)^{2}+\frac{2\Theta}{d}\,\bm\sigma(\bar{\bm\xi},\bar{\bm\xi})+\left(\bm\sigma\cdot\bm\sigma-\bm\sigma\cdot\bm\omega\right)(\bar{\bm\xi},\bar{\bm\xi}).
\end{eqnarray}
\end{subequations}
Therefore,
\begin{equation}
C_{3}=2\left(\bm\sigma\cdot\bm\sigma-\bm\sigma\cdot\bm\omega\right)(\bar{\bm\xi},\bar{\bm\xi})\,\Big|_{p}.
\label{C3}
\end{equation}
Combining Eqs.~\eqref{C1}, \eqref{C2} and (\ref{C3}), we find
\begin{equation}
C_{1}+C_{2}+C_{3}=A_{1}+A_{2}+A_{3}.
\end{equation}

\subsection{Inequality in terms of the proper section }
\label{app:ineq2}

In this subsection, we complete the proof of Proposition~\ref{prop:7}: we derive the expression~\eqref{A2:prop} of $A_{2}$ starting from the expression~\eqref{A2}.
Using Eqs.~\eqref{a xi} and \eqref{X:ei:second}, we deform the first term of Eq.~\eqref{A2} to
\begin{align}
(\nabla\otimes\bm a^{\flat})(\bar{\bm\xi},\bar{\bm\xi})&=\nabla_{\bar{\bm\xi}}\left(\bm a\cdot\bar{\bm\xi}\,\right)-\bm a\cdot\nabla_{\bar{\bm\xi}}\bar{\bm\xi}\notag\\
&=-{\cal L}_{\bm a}\left(\ln\sqrt{|\bm\xi\cdot\bm\xi|}\right).
\label{A2:2:1}
\end{align}
Then let us deform the second term of Eq.~\eqref{A2}.
Here, for any vector $\bm X\in T_{p}S$ and any vector field $\bm Y$ of $TU^{\perp}$, orthogonal to $\bar{\bm\eta}$, we can get the following equation in the absence of the vorticity $\bm\omega$ at $p$:
\begin{align}
{\cal L}_{\bm X}\left(\bm Y\cdot\bar{\bm m}\right)|_{p}&=\Big\{\nabla_{\bm X}\bm Y\cdot(\bar{\bm m}-\bar{\bm\eta})+\bm Y\cdot\nabla_{\bm X}(\bar{\bm m}-\bar{\bm\eta})+{\cal L}_{\bm X}\left(\bm Y\cdot\bar{\bm\eta}\right)\Big\}\Big|_{p}\notag\\
&=-\bm\omega(\bm X,\bm Y)|_{p}\notag\\
&=0, 
\label{XY}
\end{align}
where we have used that $S$ is the proper section in the second equality.
In other words, $\bm Y$ remains tangent to $S$ to the first order of the Taylor expansion 
around $p$ if $\bm\omega|_{p}=0$, and therefore let us define the covariant derivative of $\bm Y$ on $S$ at $p$ although $\bm Y$ is not globally a vector field of $TS$.
Let $\bm D$ be the covariant derivative on $(S,\bm h)$, and simply define the covariant derivative of a vector field $\bm Y$ of $TU^{\perp}$ at $p$ by
\begin{equation}
\bm D_{\bm X}\bm Y|_{p}:=\Big\{\nabla_{\bm X}\bm Y+\bm\chi(\bm X,\bm Y)\bar{\bm m}\Big\}\Big|_{p}.
\label{D}
\end{equation}
We use the following notation for the divergence in this subsection:
\begin{align}
\bm D\cdot\bm Y|_{p}
&=\Big\{\nabla\cdot\bm Y-\bar{\bm m}\cdot\nabla_{\bar{\bm m}}\bm Y\Big\}\Big|_{p}\notag\\
&=\Big\{\nabla\cdot\bm Y-\bar{\bm \eta}\cdot\nabla_{\bar{\bm \eta}}\bm Y\Big\}\Big|_{p}\notag\\
&=\Big\{\nabla\cdot\bm Y+\bm Y\cdot\bm a\Big\}\Big|_{p}.
\label{DY}
\end{align}
Substituting $\bm Y=\bm a$ into Eq.~\eqref{DY}, we can rewrite the second term of the expression~\eqref{A2} of $A_{2}$ to
\begin{equation}
\nabla\cdot\bm a\,|_{p}=\left(\bm D\cdot\bm a-\bm a\cdot\bm a\right)|_{p}.
\label{A2:2:2}
\end{equation}
From Eqs.~\eqref{A2:2:1} and \eqref{A2:2:2}, $A_{2}$ results in
\begin{equation}
A_{2}=\left\{-{\cal L}_{\bm a}\left(\ln\sqrt{|\bm\xi\cdot\bm\xi|}\right)+\frac{1}{d}\left(\bm D\cdot\bm a-\bm a\cdot\bm a\right)\right\}\Big|_{p}.
\label{A2:2}
\end{equation}

We further deform $A_{2}$ using the energy-momentum conservation law.
We have performed the deformation of Eq.~\eqref{app:ei} contracted with $\bm e_{i}$ as \eqref{energy:ei} in Appendix~\ref{app:ei}. 
Similarly to Eq.~\eqref{energy:ei}, 
we consider Eq.~\eqref{app:ei} contracted with $\bm a$: 
\begin{equation}
\left\{{\cal L}_{\bm a}\left(\ln\sqrt{|\bm\xi\cdot\bm\xi|}\right)+2v\,\bm a\cdot\bm b+v^{2}\bm a\cdot\bm a+\left(1-v^{2}\right)(nh)^{-1}{\cal L}_{\bm a}P\right\}\Big|_{p}=0.
\label{tang}
\end{equation}
Now $p$ is the sonic point for the radiation fluid, and the speed of the flow at $p$ is given by $1/\sqrt{d}$.
Therefore, we can rewrite the first term and the third term by Eq.~\eqref{tang}.
$A_{2}$ results in
\begin{equation}
A_{2}=\left\{\frac{1}{d}\,\bm D\cdot\bm a+\frac{2}{\sqrt{d}}~\bm a\cdot\bm b+\left(1-\frac{1}{d}\right)(nh)^{-1}{\cal L}_{\bm a}P\right\}\Big|_{p}.
\label{A2:3}
\end{equation}
Then we consider the divergence 
$\nabla\cdot\bm Y=0$ of Eq.~\eqref{energy:ei}, where $\bm Y$ is given by the projection $\perp$
of the term contracted with $\bm e_i$ in Eq.~\eqref{energy:ei} 
onto the direction orthogonal to $\bar{\bm\eta}$. 
At the point $p$, we can express $\bm Y$ in terms of
the exterior derivative $\dd_S$ on $S$ 
instead of the projection $\perp$.
Noting that $\bm Y\cdot\bm a=0$, from Eq.~\eqref{DY}, we obtain 
\begin{equation}
\bm D\cdot\left\{\dd_{S}\left(\ln\sqrt{|\bm\xi\cdot\bm\xi|}\right)+2v\,\bm b+v^{2}\bm a+(nh)^{-1}\left(1-v^{2}\right)\dd_{S}P\right\}\Big|_{p}=0.
\label{tang4}
\end{equation}
Eq.~\eqref{tang4} includes the term $\bm D\cdot\bm a$, and 
Eq.~\eqref{energy:ei} 
gives the expression of $\bm a$ in terms of the thermodynamic variables and the other part of the tensor $\bm B$.
Therefore, we can completely eliminate $\bm a$ from the expression Eq.~\eqref{A2:3} of $A_{2}$ by applying Eqs.~\eqref{energy:ei} and \eqref{tang4} to Eq.~\eqref{A2:3} at this point although the expression of $A_{2}$ gets very messy.

Finally, we take all the assumptions~\eqref{eta:m} and \eqref{dSPsv} into account to rewrite $A_{2}$.
The third term of Eq.~\eqref{A2:3} vanishes from Eq.~\eqref{dSPsv}, and evaluating $\bm b$ at $p$ gives $\bm b|_{p}=0$ from Eq.~\eqref{eta:m}.
Therefore, Eq.~\eqref{A2:3} reduces to
\begin{equation}
A_{2}=\frac{1}{d}\,\bm D\cdot\bm a\,\Big|_{p}.
\end{equation}
Eq.~\eqref{tang4} is deformed to
\begin{equation}
\left\{\Delta_{S}\left(\ln\sqrt{|\bm\xi\cdot\bm\xi|}\right)+\frac{2}{\sqrt{d}}\,\bm D\cdot\bm b+\frac{1}{d}\,\bm D\cdot\bm a+\left(1-\frac{1}{d}\right)(nh)^{-1}\Delta_{S}P\right\}\Big|_{p}=0.
\label{tang4:2}
\end{equation}
where $\Delta_{S}$ is the Laplace-Beltrami operator on $(S,\bm h)$.
We calculate the second term of Eq.~\eqref{tang4:2} by substituting $\bm Y=\bm b$ into Eq.~\eqref{DY}:
\begin{align}
\bm D\cdot\bm b\,|_{p}&=\nabla\cdot\bm b\,|_{p}\notag\\
&=\nabla\cdot\sum_{i}\bm B(\bar{\bm\xi},\bm e_{i})\bm e_{i}\,\Big|_{p}\notag\\
&=\nabla_{\mu}\Big\{(\nabla_{\nu}\bar{\eta}^{\mu})\bar{\xi}^{\nu}+\bm B(\bar{\bm\xi},\bar{\bm\xi})\bar{\xi}^{\mu}\Big\}\Big|_{p}\notag\\
&=[\nabla_{\mu},\nabla_{\nu}]\bar{\eta}^{\mu}~\bar{\xi}^{\nu}\notag\\
&={\rm Ric}(\bar{\bm\eta},\bar{\bm\xi})|_{p}\notag\\
&={\rm Ric}(\bar{\bm m},\bar{\bm\xi})|_{p}.
\end{align}
Applying Eq.~\eqref{tang4:2} to Eq.~\eqref{A2:3}, we obtain the expression~\eqref{A2:prop} of $A_{2}$ in terms of the proper section  $S$.


\end{document}